\documentclass[letterpaper,onecolumn,10pt,accepted=2022-06-01]{quantumarticle}

\pdfoutput=1

\usepackage{geometry}
\usepackage[utf8]{inputenc}
\usepackage[english]{babel}
\usepackage[T1]{fontenc}
\usepackage{amsmath,amssymb,amsthm}
\usepackage{authblk}
\usepackage{tabu}
\usepackage{dsfont}
\usepackage{braket}
\usepackage{cancel}
\usepackage{graphicx}
\usepackage{xcolor}
\usepackage{float}
\usepackage{array, makecell}
\usepackage{multirow}
\usepackage{setspace}
\usepackage[linewidth=1pt]{mdframed}

\usepackage[numbers]{natbib}
\usepackage[allcolors=blue, colorlinks=true,hyperindex=true]{hyperref}

\DeclareMathOperator{\EV}{\mathbb{E}}

\newtheorem*{theorem*}{Theorem}
\newtheorem*{corollary*}{Corollary}
\newtheorem{defn}{Definition}
\newtheorem{lemma}{Lemma}

\newcolumntype{C}{>{$}c<{$}}
\newcolumntype{L}{>{$}l<{$}}
\newcolumntype{R}{>{$}r<{$}}
\setcellgapes{3.5pt}

\newcommand*{\mvcenter}[1]{\vcenter{\hbox{$\displaystyle #1$}}}

\newcommand{\Id}{\mathds{1}}
\newcommand{\vect}[1]{\boldsymbol{#1}}
\newcommand{\av}{{\vect{a}}}
\newcommand{\bv}{{\vect{b}}}
\newcommand{\cv}{{\vect{c}}}
\newcommand{\dv}{{\vect{d}}}
\newcommand{\zv}{{\vect{z}}}
\newcommand{\xv}{{\vect{x}}}
\newcommand{\yv}{{\vect{y}}}
\newcommand{\uv}{{\vect{u}}}
\newcommand{\vv}{{\vect{v}}}
\newcommand{\vparam}{\vgamma,\vbeta}

\newcommand{\Am}{Y}
\newcommand{\M}{{\rm m}}

\newcommand{\vgamma}{\vect{\gamma}}
\newcommand{\vbeta}{\vect{\beta}}
\newcommand{\gbbraket}[1]{\braket{\gamma,\beta|#1|\gamma,\beta}}
\newcommand{\bgbbraket}[1]{\braket{\vgamma,\vbeta|#1|\vgamma,\vbeta}}

\renewcommand{\P}[1]{{[#1]}}

\newcommand{\EVJ}[1]{\EV_J\left[#1\right]}

\newcommand{\I}{\mathcal{I}}

\newcommand{\Ac}{B}

\newcommand{\Suv}{S_{\uv,\vv}}
\newcommand{\Suvxy}{S_{\uv,\vv,\xv,\yv}}
\newcommand{\suv}{s_{\uv,\vv}}
\newcommand{\buv}{b_{\uv,\vv}}
\newcommand{\luv}{\ell_{\uv,\vv}}

\newcommand{\Qf}[1]{Q_{#1}^{n_{#1}} Q_{\bar{#1}}^{n_{\bar{#1}}}}
\newcommand{\Ff}[1]{|F_{#1}|^{-t_{#1}}}

\newcommand{\gf}[2]{(g_{#1#2}^{n_{#2}} g_{#1\bar{#2}}^{n_{\bar{#2}}})^{t_{#1}}}
\newcommand{\gt}[2]{g_{#1#2}^{t_{#1}}}
\newcommand{\Xf}[1]{ X_{#1}^{n_{#1}}   X_{\bar{#1}}^{n_{\bar{#1}}}   }

\begin{document}
\title{The Quantum Approximate Optimization Algorithm and the Sherrington-Kirkpatrick Model at Infinite Size}

\author[1,2]{Edward Farhi}
\author[2]{Jeffrey Goldstone}
\author{Sam Gutmann}
\author[1,3]{Leo Zhou}

\affil[1]{\normalsize Google Inc., Venice, CA 90291, USA}
\affil[2]{Center for Theoretical Physics, Massachusetts Institute of Technology, Cambridge, MA 02139, USA}
\affil[3]{Department of Physics, Harvard University, Cambridge, MA 02138, USA}

\maketitle

\begin{abstract}
\normalsize
\onehalfspacing
The Quantum Approximate Optimization Algorithm (QAOA) is a general-purpose algorithm for combinatorial optimization problems whose performance can only improve with the number of layers $p$. While QAOA holds promise as an algorithm that can be run on near-term quantum computers, its computational power has not been fully explored. In this work, we study the QAOA applied to the Sherrington-Kirkpatrick (SK) model, which can be understood as energy minimization of $n$ spins with all-to-all random signed couplings.
There is a recent classical algorithm~\cite{Montanari2018} that, assuming a widely believed conjecture, can efficiently find an approximate solution for a typical instance of the SK model to within $(1-\epsilon)$ times the ground state energy.
We hope to match its performance with the QAOA.

Our main result is a novel technique that allows us to evaluate the typical-instance energy of the QAOA applied to the SK model.
We produce a formula for the expected value of the energy, as a function of the $2p$ QAOA parameters, in the infinite size limit that can be evaluated on a computer with $O(16^p)$ complexity.
We evaluate the formula up to $p=12$, and find that the QAOA at $p=11$ outperforms the standard semidefinite programming algorithm.
Moreover, we show concentration: With probability tending to one as $n\to\infty$, measurements of the QAOA will produce strings whose energies concentrate at our calculated value.
As an algorithm running on a quantum computer, there is no need to search for optimal parameters on an instance-by-instance basis since we can determine them in advance.
What we have here is a new framework for analyzing the QAOA, and our techniques can be of broad interest for evaluating its performance on more general problems where classical algorithms may fail. 
\end{abstract}

\clearpage

\section{Introduction}
The Quantum Approximate Optimization Algorithm~\cite{QAOA}, QAOA, consists of a shallow depth quantum circuit with $p$ layers and $2p$ parameters.  It is designed to find approximate solutions to combinatorial search problems, and like simulated annealing can be applied to almost any problem.  As $p$ increases, performance can only improve.  Even at $p=1$, worst case performance guarantees have been established~\cite{QAOA,QAOAE3Lin2}. These beat random guessing but not the best classical algorithms.  Aside from worst case we can ask how the QAOA performs on typical instances where the problem instances are drawn from a specified distribution. For example, for the combinatorial search problem E3LIN2, a worst case performance guarantee was established but the typical performance is better~\cite{QAOAE3Lin2}. For similar results on other problems see Lin and Zhu~\cite{QAOALinZhu}.

A key insight for typical instances, drawn from a specified distribution, is the landscape independence~\cite{QAOAFixedParam} of the expected value of the cost function in the parameter dependent quantum state: Given an instance of a combinatorial search problem and a set of control parameters, the QAOA quantum state depends both on the instance and the parameters.  However for large systems,  the expected value of the cost function in the quantum state depends on the parameters but not the instance, up to finite size effects. This means that for each value of $p$,  the optimal parameters are the same for typical instances coming from the distribution.  Although the optimal parameters may be known, one still needs to run the quantum computer for each instance to find a string which achieves the optimal cost function value attainable at a given $p$.

In this paper, we introduce a novel technique for evaluating the typical-instance energy of the QAOA applied to a combinatorial search problem for arbitrary fixed $p$ as the problem size goes to infinity.  
We apply the QAOA to the Sherrington-Kirkpatrick model.
This model can be seen as a combinatorial search problem on a complete graph with random couplings.  For typical instances, the value of the lowest energy is known~\cite{Parisi, Panchenko, CrisantiRizzo2002, SchmidtThesis2008}.
Assuming a widely believed conjecture about the model, there is also a recent polynomial-time classical algorithm for finding a string whose energy is between $(1-\epsilon)$ times the lowest energy and the lowest energy with high probability~\cite{Montanari2018}.
What we have accomplished is a method for calculating, in advance, what energy the QAOA will produce for given parameters at any fixed $p$ in the infinite size limit.
We also show concentration,
i.e.,  with probability tending to one as $n\to\infty$,
measurements of the QAOA will produce strings whose energies concentrate at our calculated value. 
This concentration result implies landscape independence for this model, similar to what was found for bounded degree graphs in \cite{QAOAFixedParam}.
This means that optimal parameters found for one large instance will also be good for other large instances.
The complexity of our calculation scales as $O(16^p)$, and the answer is evaluated on a computer.
Without much trouble, we have found optimal parameters up to $p=8$, and further evaluated the performance up to $p=12$ at extrapolated parameters.
Notably, our results show that the QAOA surpasses the standard semidefinite programming algorithm at $p=11$.

The paper is organized as follows. We first review the QAOA (Section \ref{sec:QAOA}) and the Sherrington-Kirkpatrick model (Section \ref{sec:SK}).
The results of the paper are summarized in Section \ref{sec:results}, where we produce a formula to calculate the typical-instance energy achieved by the QAOA for any $p$ with any given parameters, and show that the measured energy concentrates at the calculated value.
We also evaluate the formula up to $p=12$.
In Section \ref{sec:p=1}, we demonstrate how to calculate the typical-instance energy for the QAOA at $p=1$ for any finite problem size. 
In Section \ref{sec:general-p}, we prove the formula for general $p$ in the infinite size limit.
We provide further discussion in Section \ref{sec:discussion}.

\section{The Quantum Approximate Optimization Algorithm (QAOA)\label{sec:QAOA}}

We start by reviewing the Quantum Approximate Optimization Algorithm (QAOA).
Given a classical cost function $C(\zv)$ defined on $n$-bit strings $\zv=(z_1,z_2,\ldots, z_n) \in \{+1, -1\}^n$, the QAOA is a quantum algorithm that aims to find a string $\zv$ such that $C(\zv)$ is close to its absolute minimum.
The cost function $C$ can be written as an operator that is diagonal in the computational basis, defined as
\begin{equation}
C\ket{\zv} = C(\zv)\ket{\zv}.
\end{equation}
We introduce a unitary operator that depends on $C$ and a parameter $\gamma$
\begin{equation} \label{eq:UC}
U(C,\gamma) = e^{-i\gamma C}.
\end{equation}
Additionally, we introduce the operator
%\vspace{-10pt}
\begin{equation}
B = \sum_{j=1}^n X_j,
\end{equation}
where $X_j$ is the Pauli $X$ operator acting on qubit $j$, and an associated unitary operator that depends on a parameter $\beta$
%\vspace{-10pt}
\begin{equation} \label{eq:UB}
U(B,\beta) = e^{-i\beta B} = \prod_{j=1}^n e^{-i\beta X_j}.
\end{equation}
In the QAOA circuit, we initialize the system of qubits in the state
\begin{equation}
\ket{s} = \ket{+}^{\otimes n} =  \frac{1}{\sqrt{2^n}} \sum_{\zv}\ket{\zv},
\end{equation}
and alternately apply $p$ layers of $U(C,\gamma)$ and $U(B,\beta)$.
Let $\vect\gamma = \gamma_1,\gamma_2,\ldots,\gamma_p$ and $\vect\beta = \beta_1,\beta_2,\ldots,\beta_p$; then
the QAOA circuit prepares the following state
\begin{equation} \label{eq:wavefunction}
\ket{\vect\gamma, \vect\beta} = U(B,\beta_p) U(C,\gamma_p) \cdots U(B,\beta_1) U(C, \gamma_1)\ket{s}.
\end{equation}
For a given cost function $C$, the associated QAOA objective function is
\begin{equation}
\braket{\vgamma,\vbeta|C|\vgamma,\vbeta}. \label{eq:QAOA_obj}
\end{equation}
By  measuring the quantum state $\ket{\vgamma,\vbeta}$ in the computational basis, one will find a bit string $\zv$ such that $C(\zv)$ is near \eqref{eq:QAOA_obj} or better.
Different strategies have been developed to find optimal $(\vgamma,\vbeta)$ for any given instance~\cite{QAOAFixedParam, ZhouQAOA,Crooks2018}.
In this paper, however, we will show how to find, in advance, the optimal parameters of the QAOA at fixed $p$ for typical instances of the Sherrington-Kirkpatrick model.

\vspace{-5pt}
\section{The Sherrington-Kirkpatrick (SK) model\label{sec:SK}}

In this paper, we apply the QAOA to instances of the Sherrington-Kirkpatrick model, which is central to many important results in the theory of spin glass and disordered systems (see \cite{Panchenko} for a review).
This model describes a classical spin system with all-to-all couplings between the $n$ spins.
The classical cost function is
\begin{equation}
\label{eq:SK-model}
C(\zv) = \frac{1}{\sqrt{n}}  \sum_{j<k} J_{jk} z_j z_k.
\end{equation}
Each instance is specified by the $J_{jk}$'s which are independently chosen from a distribution with mean 0 and variance 1.
For example they could be $+1$ or $-1$ each with probability 1/2, or drawn from the standard normal distribution. 
In this paper, we assume $J_{jk}$ comes from a symmetrical distribution (i.e., $J_{jk}$ and $-J_{jk}$ have the same distribution), although we believe generalization to any subgaussian distribution is not too difficult.
%(They can even be drawn from an asymmetric distribution, such as one describing MaxCut on Erd\H{o}s-R\'enyi graphs with $J_{jk}=(A_{jk}-\EV[A_{jk}])/\sqrt{\text{Var}(A_{jk})}$, where $A_{jk}$ forms the adjacency matrix of each graph~\cite{DMS17,Montanari2018}.)
As $n$ goes to infinity, the classical results as well as our quantum results are independent of which distribution the $J_{jk}$'s come from.  In a celebrated result, Parisi calculated the lowest energy of \eqref{eq:SK-model} for typical instances as $n$ goes to infinity. The first step is to consider the partition function
\begin{equation}
				\sum_{\zv} \exp ( - C(\zv)/T )
\end{equation}
which is dominated by the low energy configurations for low temperature $T$.  Averaging this quantity over the $J_{jk}$'s is easy but gives the wrong description at low temperature.  One needs to first take the log and then average:
\begin{equation}
				\frac{T}{n} \EV_J \Big[\log \Big( \sum_{\zv} \exp ( - C(\zv)/T ) \Big)\Big].
\end{equation}
In the limit of $n\to\infty$, this quantity has been proven to exist, although evaluating this quantity is extremely challenging.  Using the method of replica symmetry breaking, Parisi~\cite{Parisi} provides a formula (for proof, see~\cite{Panchenko}) that, when numerically evaluated~\cite{CrisantiRizzo2002, SchmidtThesis2008, Montanari2018}, shows for typical instances,
\begin{equation}
	\lim_{n\to\infty} \min_{\zv} \frac{C(\zv)}{n} = - 0.763166\ldots.
	\label{eq:Parisi}
\end{equation}
We note that this model has a phase transition at $T = 1$, and that at this temperature,  $C/n$ is $-0.5$ for typical instances.

The result just given tells us what the lowest energy of the SK model is, but it does not tell us how to find the associated configuration.
How hard is this? 
Simulated annealing does rather well in getting a string whose energy is close to the minimum.
For example running simulated annealing starting at a temperature above the phase transition, say $T=1.3$, and bringing it down to near $T=0$ in 10 million steps on a 10000 bit instance achieves $-0.754$. 
It is not believed that simulated annealing can go all the way to the lowest energy configuration~\cite{ParisiPrivateComm}. 
We can also ask how zero-temperature simulated annealing performs.
Here the update rule is flip a spin if it lowers the energy, and don't flip if it does not. 
This algorithm comes down to around $-0.71$.

Another natural approach to solve this problem is a convex relaxation method such as semidefinite programming. The simplest such method is called spectral relaxation, where the solution is obtained by taking the lowest eigenvector of the $n\times n$ matrix $\vect{J}$ (whose elements are given by $J_{jk}$), and rounding each entry to $\pm 1$ based on its sign. This yields $C/n=-2/\pi + o(1)\approx-0.6366$~\cite{ALR87}, where $o(1)$ is a number that vanishes as $n\to\infty$. The standard semidefinite programming algorithm with ``sign rounding'' yields the same result~\cite{MS16,BKW19}.

Recently, assuming a widely believed yet unproven conjecture that the SK model has ``full replica symmetry breaking,'' Montanari gave an algorithm~\cite{Montanari2018} that finds a string whose energy is below $(1- \epsilon)$
times the lowest energy for typical instances. The run time of this algorithm is $c(\epsilon) n^2$, and the function $c(\epsilon)$ is an inverse polynomial of $\epsilon$.

\section{Summary of Results}
\label{sec:results}
We analyze how the QAOA performs on the SK model by evaluating
\begin{equation}
\EVJ{\bgbbraket{C/n}} \quad \text{and} \quad \EVJ{\bgbbraket{(C/n)^2}}
\end{equation}
in the infinite size $n\to\infty$ limit.
Note that both $C$ and $\ket{\vparam}$ depend on $J$.
Our main result, which will be proved in Section~\ref{sec:general-p}, is the following:
\begin{theorem*}\label{thm}
For any $p$ and any parameters $(\vparam)$, we have
\begin{align} \label{eq:Vp-def}
 \lim_{n\to\infty} \EVJ{\bgbbraket{C/n}}  &= V_p(\vparam),
\end{align}
where $V_p(\vparam)$ has an explicit formula that we give below in Eq.~\eqref{eq:Vp-formula}.
Furthermore, we have
\begin{align} \label{eq:second-moment}
  \lim_{n\to\infty} \EVJ{\bgbbraket{(C/n)^2}} =  [V_p(\vparam)]^2.
\end{align}
\end{theorem*}

Note the above theorem easily implies the following corollary which shows that the QAOA exhibits strong concentration properties:

\begin{corollary*}[Concentration]
At any fixed $p$ and for any parameters $(\vparam)$, the energy produced by the QAOA on almost every instance of the SK model \emph{concentrates} at $V_p(\vparam)$ in the limit as $n\to\infty$.
Specifically, we have \emph{concentration over instances}, i.e., for every $\epsilon > 0$,
\begin{align} \label{eq:concentration-instances}
\mathbb{P}_J \left(\Big|\bgbbraket{C/n} - V_p(\vparam)\Big| > \epsilon \right) \to 0
 \qquad \text{as} \qquad n\to\infty.
\end{align}
We also have \emph{concentration over measurements}, i.e., for every $\epsilon>0$,
\begin{align} \label{eq:concentration-measure}
%\bgbbraket{(C/n)^2} - \bgbbraket{C/n}^2  \to 0 \qquad \text{as} \qquad n\to\infty.
\EV_J \Big[\mathbb{P}_{\zv \textnormal{ measured from } \ket{\vparam}}\left( \Big|C(\zv)/n - \bgbbraket{C/n}\Big| > \epsilon \right)\Big] \to 0
\qquad \text{as} \qquad n\to\infty.
\end{align}
\end{corollary*}
\begin{proof}[Proof of Corollary]
Combining \eqref{eq:Vp-def} and \eqref{eq:second-moment}, we have
\begin{align}
\lim_{n\to\infty}\left(\EV_J[\bgbbraket{(C/n)^2}] - \EV_J^2[\bgbbraket{C/n}]\right) = 0.
\label{eq:moment-diff}
\end{align}
Since for any operator $O$,
\begin{align} \label{eq:op-ineq}
\braket{O}^2 \le \braket{O^2},
\end{align}
we have
\begin{align}
\EV_J[\bgbbraket{C/n}^2] - \EV_J^2[\bgbbraket{C/n}]
&\le
	\EV_J[\bgbbraket{(C/n)^2}] - \EV_J^2[\bgbbraket{C/n}].
\end{align}
Note that as $n\to\infty$, the right hand side vanishes, which implies
\begin{align}	\label{eq:concentrate-exp}
\lim_{n\to\infty} \left(
\EV_J[\bgbbraket{C/n}^2] - \EV_J^2[\bgbbraket{C/n}]
\right) = 0.
\end{align}
By Chebyshev's inequality, this implies \eqref{eq:concentration-instances}, which proves landscape independence~\cite{QAOAFixedParam} for the SK model.

To show concentration over measurements for fixed $J$, we subtract \eqref{eq:concentrate-exp} from \eqref{eq:moment-diff} to give us
\begin{align}
\EV_J\Big[\bgbbraket{(C/n)^2}-\bgbbraket{C/n}^2\Big] &\longrightarrow 0 \qquad \text{as} \qquad n\to\infty.
\label{eq:concentrate-string}
\end{align}
By \eqref{eq:op-ineq} for any $J$ 
\begin{equation}
\bgbbraket{(C/n)^2} - \bgbbraket{C/n}^2 \ge 0,
\end{equation}
Then \eqref{eq:concentrate-string}
shows that this variance vanishes typically as $n\to\infty$.
This implies \eqref{eq:concentration-measure} by Chebyshev's inequality.
\end{proof}

This corollary means that as $n\to\infty$, for typical instances of the SK model, upon applying the QAOA and measuring in the computational basis, we will obtain a string $\zv$ that has energy $C(\zv)/n$ close to $\bgbbraket{C/n}$, which is itself close to $V_p(\vparam)$.

\paragraph{Formula for $V_p(\vparam)$}---
We now provide instructions on how to evaluate $V_p(\vparam)$.
Start by defining
\begin{equation}
A=\{(a_1,\ldots, a_p, a_{-p}, \ldots, a_{-1}): a_{\pm j} = +1, -1\}
\end{equation}
as the set of all configurations.
We partition into $A=A_1\cup\cdots \cup A_p \cup A_{p+1}$, where
\begin{equation}
A_\ell = \{\av:  a_{-k}  = a_{k} \text{ for } p-\ell+1 < k \le p, \text{ and } a_{-p+\ell- 1} = - a_{p-\ell+1}\}
\quad
\text{for}
\quad
1\le \ell \le p,
\end{equation}
and
\begin{equation}
A_{p+1} = \{\av : a_{-k} = a_{k}  \text{ for } 1\le k \le p\} \,. 
\end{equation}
We define a $*$ operation which transforms configuration $\av$ into $\av^*$ given by
\begin{equation}
a^*_r=a_r a_{r+1}\cdots a_p
\qquad \text{and} \qquad
a^*_{-r} = a_{-r} a_{-r-1}\cdots a_{-p}
\quad
\text{for } 1\le r\le p \,.
\end{equation}
For any $\av, \bv \in A$, let
\begin{gather}
Q_\av = \prod_{j=1}^p
	(\cos\beta_j)^{1 + (a_j + a_{-j})/2}
	(\sin\beta_j)^{1 - (a_j + a_{-j})/2}
	(i)^{(a_{-j} - a_{j})/2}, \\
\Phi_\av = \sum_{r=1}^p \gamma_r (a^*_r - a^*_{-r}), \qquad
\text{and} \qquad
X_\bv = Q_\bv \exp\Big({-}\frac12{  \sum_{\av\in A_{p+1}} Q_\av \Phi^2_{\av\bv}}\Big).
\end{gather}
Furthermore, for any $\av, \bv \in A\setminus A_{p+1}$, let
\begin{gather}
\Delta_{\av, \bv} = \frac12 (\Phi^2_{\bar\av\bv} - \Phi^2_{\av\bv})
\end{gather}
where the $\bar{\av}$ operation takes configuration $\av \in A_\ell$ to $\bar\av \in A_\ell$ for $1 \le \ell \le p$ via
\begin{align}
\bar{a}_{\pm r} = \begin{cases}
a_{\pm r}, & r \neq p-\ell+1 \\
-a_{\pm r} & r = p-\ell+1
\end{cases}
\,.
\end{align}

We then bipartition $\Ac=A\setminus A_{p+1}$ into two sets $D \cup D^c$, such that if $\bv \in D$ then $\bar\bv \in D^c$.
Note $|D|=|B|/2=(2^{2p}-2^p)/2$.
For concreteness, we can use the convention where all configurations of $D$ are chosen so that their first $p$ elements have even parity.
Furthermore, for each element $\bv\in D$ we assign an index $j(\bv) \in \{1,2,\ldots, |D|\}$, such that if $j(\bv) \le j(\bv')$ then $\Phi_{\bar\bv\bv'}^2 = \Phi^2_{\bv\bv'}$ (and we say ``$\bv$ plays well with $\bv'$'').
This index can be efficiently assigned since $\bv\in A_\ell$ plays well with $\bv'\in A_{\ell'}$ if $\ell \le \ell'$ (see Lemma~\ref{lem:ordering} in Section~\ref{sec:1=1}).
With this index assignment, we have $\Delta_{k,j}=0$ for $k\le j$.

Next, for every element of $\uv\in D$ we assign a value $W_\uv$ which can be iteratively obtained as follows:
Start with $W_{|D|} = X_{|D|}$, and for $j=|D|-1, \ldots, 2, 1$ compute
\begin{equation} \label{eq:W-iter}
W_j = X_j \exp\left({\textstyle \sum_{k=j+1}^{|D|} W_k \Delta_{k,j} }\right).
\end{equation}
Then for other elements $\uv \not\in D$ we assign
\begin{equation}
W_\uv = Q_\uv 
\quad \text{when}\quad
 \uv\in A_{p+1}, 
\qquad
\text{ and }
\qquad 
W_{\bar\uv} = - W_\uv 
\quad \text{when}\quad
\bar\uv \in D^c.
\end{equation}

Finally, $V_p(\vparam)$ is obtained from all the $W_\uv$'s via the following formula:
\begin{equation} \label{eq:Vp-formula}
V_p(\vparam) = \frac{i}{2} \sum_{r=1}^p \gamma_r \left[\sum_{\uv\in A} (u^*_r + u^*_{-r})  W_\uv\right] \left[\sum_{\vv\in A} (v^*_r-v^*_{-r})W_\vv\right]\,.
\end{equation}
The memory complexity of the above computation is at most ${O}(4^p)$ for storing the vector of $W_\uv$'s which has dimension $|D|=(4^p-2^p)/2$.
Computing all the $W_{\uv}$'s requires evaluating \eqref{eq:W-iter} $|D|$ times, each involving a sum of at most $|D|$ terms.
Hence, we only need at most $O(|D|^2) = O(16^p)$ time to compute the $W_\uv$'s.
Once we have an array of $W_\uv$ in memory, computing $V_p$ takes at most $O(p \,  4^p)$ extra time as we simply add up $W_\uv$ for some $p$ pairs of subsets of $A$ and multiply the results.
The time complexity for computing $V_p$ is thus at most $O(16^p)$.

\paragraph{Numerical evaluation and optimization}---
Here, we give our results optimizing
\begin{equation}
V_p(\vgamma,\vbeta) = \lim_{n\to\infty} \EV[\bgbbraket{C/n}],
\end{equation}
with respect to the QAOA parameters $\vgamma$ and $\vbeta$.
For a given $p$, we denote the optimal value as
\begin{align}
\bar{V}_p = \min_{\vgamma,\vbeta} V_p(\vgamma, \vbeta).
\end{align}
Using a laptop to calculate $V_p(\vparam)$ to floating point precision, and applying a heuristic for optimizing QAOA parameters~\cite{ZhouQAOA}, we are able to find optimal parameters up to $p=8$.
Furthermore, by performing the computation using a GPU on the Harvard FASRC Cannon cluster, we are able to calculate $V_p(\vparam)$ up to $p=12$.

In Figure~\ref{fig:Cexp}(a), we plot the values of the energy in the infinite size limit, $V_p(\vgamma,\vbeta)$, at the parameters we have found.
The values we have found for $9\le p \le 12$ are not optimized but simply evaluated at some parameters extrapolated from those at lower $p$.
These values serve as upper bounds on the value of $\bar{V}_p$, since we are not always certain that the parameters we found are the best possible.
It is shown to decrease steadily as $p$ increases, but it is unclear where or how fast it will asymptote.
Nevertheless, the results show that the QAOA surpasses the standard semi-definite programming algorithm at $p=11$, where $\bar{V}_{11} \le -0.6393 < -2/\pi$.
The known minimum energy of $-0.763166...$ from the Parisi ansatz as in \eqref{eq:Parisi} is shown as a red line.
To confirm that our infinite size limit calculation gives reasonable predictions in the finite size case, we simulate the QAOA on 30 randomly tossed instances of the SK model with $n=26$ spins and $J_{jk}$ drawn from the standard normal distribution.
In Figure~\ref{fig:Cexp}(b), we plot the expectation value of the energy $\braket{C/n}$ for these instances, at the same optimized QAOA parameters used for the infinite size limit.
It appears to give good agreement with the values in Figure~\ref{fig:Cexp}(a), up to $O(1/\sqrt{n})$ finite size effects.

We have found that the optimal QAOA parameters for the SK model in the infinite size limit appear to have a pattern, similar to what was found in \cite{ZhouQAOA, Crooks2018}. As an example, we plot the optimal parameters for $p=8$ in Figure~\ref{fig:param}.
The full list of parameters and $\bar{V}_p$ we found for $1\le p \le 8$ are also given in Table~\ref{paramtable}.
As we progress through the QAOA circuit, the optimal parameter $\gamma_i$ appears to increase monotonically, and the parameter $-\beta_i$ appears to decrease monotonically.
We have exploited this pattern to obtain good guesses of parameters for $9\le p \le 12$, which are listed in Table~\ref{paramtable2}.

%We have attempted to fit our data of $\bar{V}_p$ for $p\le 8$ to an exponential-type function $\bar{V}_p = a\exp(-p^b/c)+d$ as well as a power-law $\bar{V}_p = a/|p+c|^b+d$, since these seem to be the simplest and most reasonable models.
%Both fits have 4 free parameters $(a,b,c,d)$ that include a variable offset.
%The results are shown in Figure~\ref{fig:fit}.
%The fit to the exponential-type function suggests that $\bar{V}_p$ may asymptote to a value above the loweswt energy predicted by Parisi. %, but it appears to underestimate the performance achievable by QAOA for $p\ge 9$ as shown in the inset of Fig.~\ref{fig:Cexp}(a).
%On the other hand, the power-law fit indicates that $\bar{V}_p$ will converge to a value that is $-0.778\pm0.037$, which is consistent with the Parisi value.
%The inset of Figure~\ref{fig:fit} shows the sum of squared errors of the two fitting models with various fixed $d$, and the power-law model appears to give a slightly better fit.
%%It also does not underestimate the performance of QAOA, and hence may be a better predictor of the asymptotic behavior of $\bar{V}_p$.
%The fit parameters we obtain for the power-law fit are $(a,b,c,d) = (1.10,0.817,1.81,-0.778) \pm (0.29,0.191,0.50,0.037)$, where the errors are 95\% confidence intervals.

\begin{figure}[tb]
\centering
\includegraphics[height=6cm]{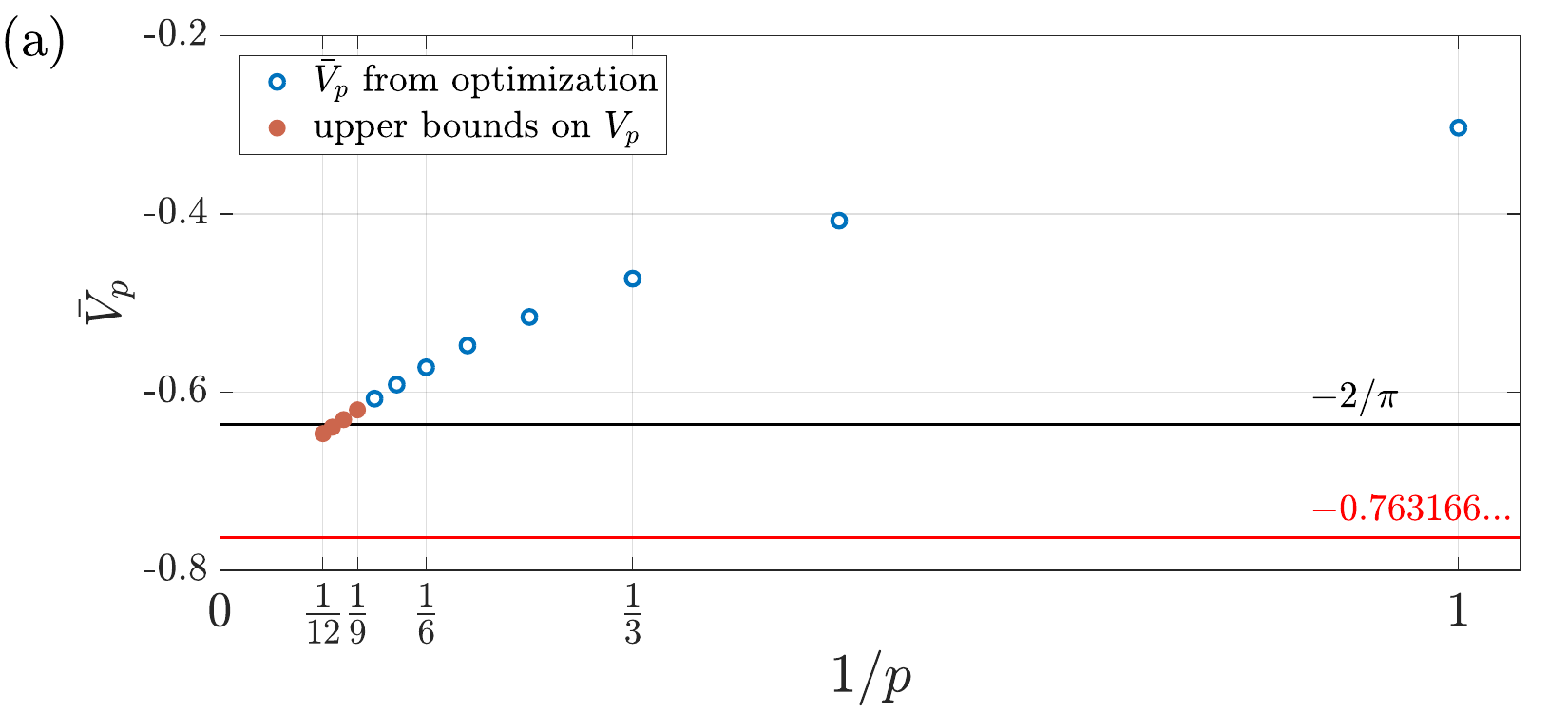}\\
\includegraphics[height=6cm]{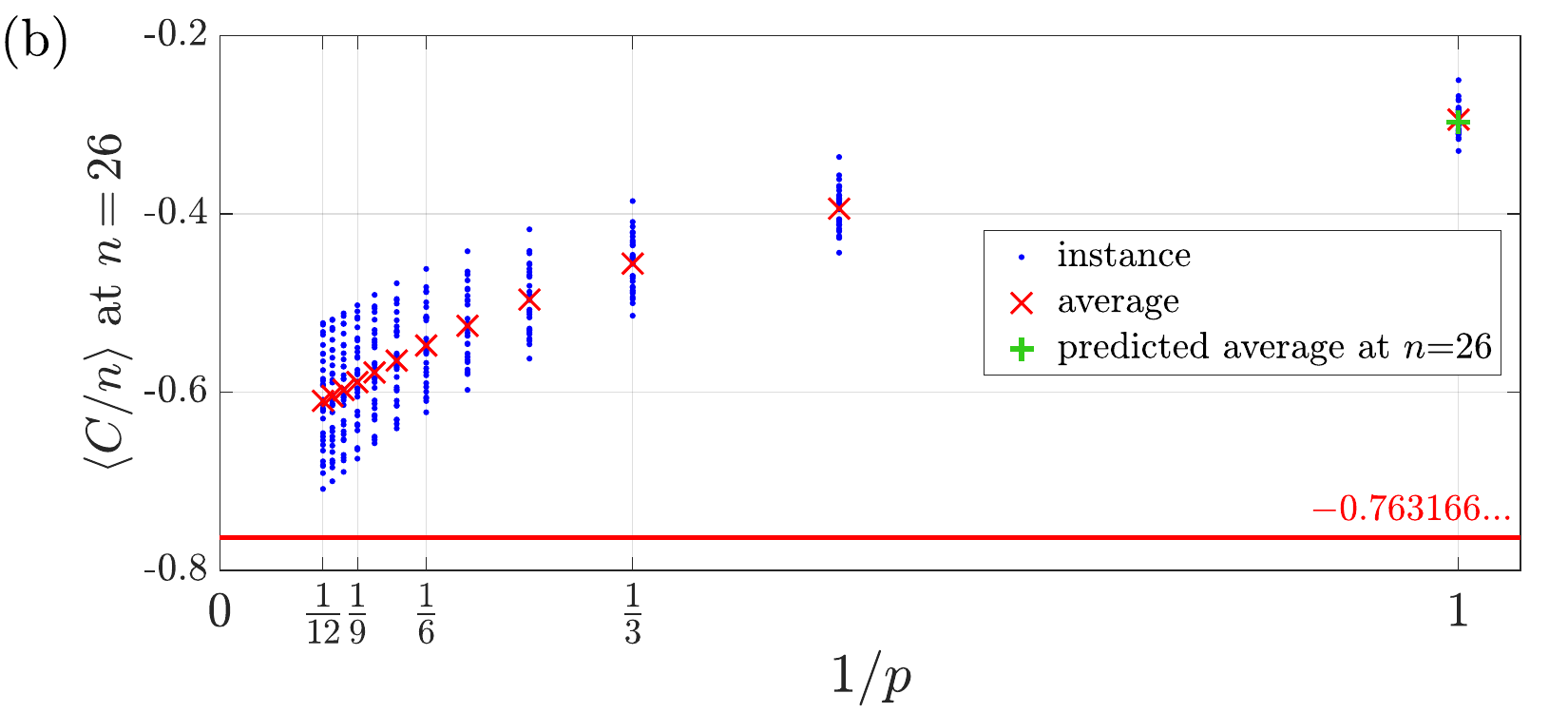}
\vspace{-5pt}
\caption{\label{fig:Cexp}
(a) Our results for $\bar{V}_p =\min_{\vparam} V_p(\vparam)$ as a function of $p$. The values for $p\le 8$ are based on optimizing $V_p(\vparam)=\lim_{n\to\infty} \EVJ{\braket{C/n}}$. The values for $9\le p \le 12$ are from evaluating $V_p(\vparam)$ at some parameters extrapolated from those at lower $p$ but not further optimized. 
Note at $p=11$, QAOA is able to surpass $-2/\pi\approx 0.6366$, which is the performance of the standard semidefinite programming algorithm~\cite{Montanari2018,BKW19}.
The full set of values of $V_p$ and the parameters are given in Table~\ref{paramtable} and \ref{paramtable2}.
(b) The expectation value of $C/n$ for 30 randomly tossed instances of 26-spin Sherrington-Kirkpatrick model with $J_{jk}$ drawn from the standard normal distribution.
The spread in values over the instances is a finite $n$ effect, since our concentration results show that the variance over instances goes to 0 as $n\to\infty$.
The average values over instances show good agreement with those above in (a).
The $p$\,=\,1 predicted average at $n=26$ from \eqref{eq:C-p1-n}, which evaluates to $-0.29726$, is also plotted.}
\end{figure}

\vspace{-10pt}
\begin{figure}[bt]
\centering
\includegraphics[height=4.3cm]{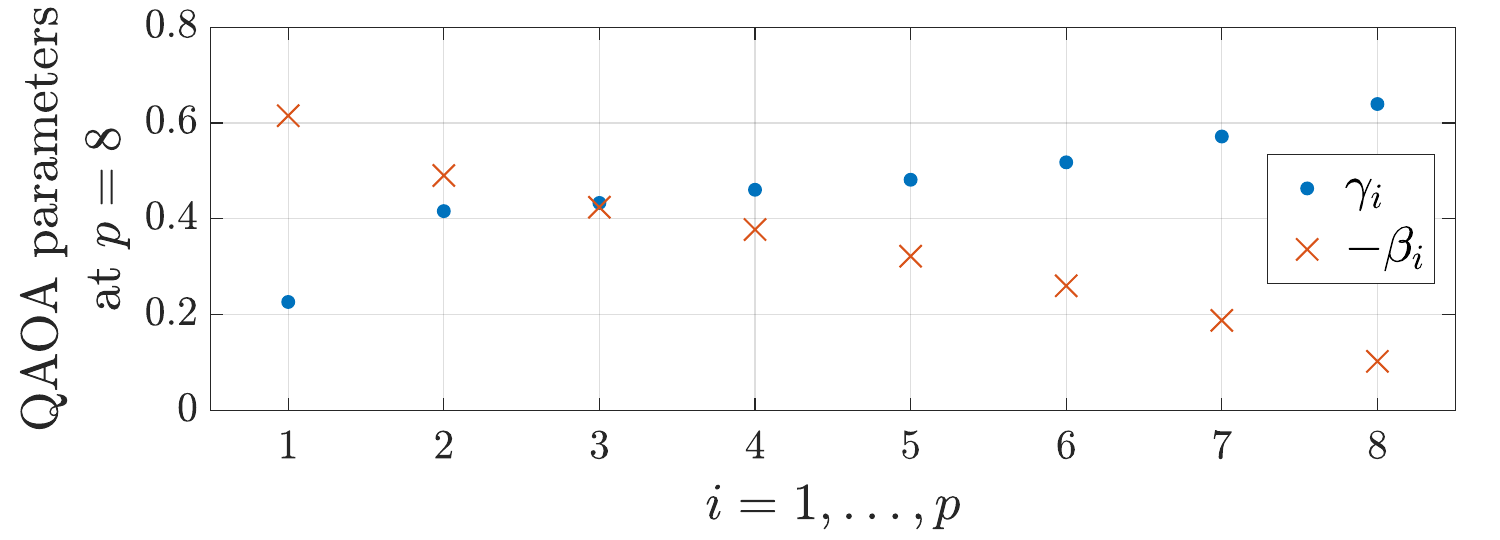}
\vspace{-5pt}
\caption{The optimal parameters found for $p=8$. 
We believe these to be globally optimal, based on the fact they appear repeatedly as the best local minimum when optimizing from $10^4$ starting points uniformly randomly drawn from the range $\gamma_i\in[-2,2]$ and $\beta_i\in[-\pi/4,\pi/4]$.
\label{fig:param}}
\end{figure}

\clearpage

\begin{table}[H]
\centering
\makegapedcells
\begin{tabular}{c|c|c|c}
$p$ & $\bar{V}_p$ & optimized $\vgamma$ & optimized $-\vbeta$ \\
\hline \hline
\multirow{2}{*}{1} & $-1/\sqrt{4e}\approx$ & \multirow{2}{*}{1/2} & \multirow{2}{*}{$\pi/8$} \\ 
& $-0.3033$ & & \\ \hline
2 & $-0.4075$ & 0.3817, 0.6655 & 0.4960, 0.2690 \\ \hline
3 & $-0.4726$ & 0.3297, 0.5688, 0.6406 & 0.5500, 0.3675, 0.2109\\ \hline
4 & $-0.5157$ & 0.2949, 0.5144, 0.5586, 0.6429 & 0.5710, 0.4176, 0.3028, 0.1729 
\\ \hline
\multirow{2}{*}{5} & \multirow{2}{*}{$-0.5476$ }
& 0.2705, 0.4803, 0.5074 & 0.5899, 0.4492, 0.3559 \\
& & 0.5646, 0.6397  & 0.2643, 0.1486
\\ \hline
\multirow{2}{*}{6} & \multirow{2}{*}{$-0.5721$} 
& 0.2528, 0.4531, 0.4750
& 0.6004, 0.4670, 0.3880 \\
& & 0.5146, 0.5650, 0.6392 & 0.3176, 0.2325, 0.1291
\\ \hline
 &
& 0.2383, 0.4327, 0.4516,
 & 0.6085, 0.4810, 0.4090, 
 \\
7 & $-0.5915$
 & 0.4830, 0.5147,
 & 0.3535, 0.2857,
 \\
 &
 & 0.5686, 0.6393
 & 0.2080, 0.1146
\\ \hline
 & 
 & 0.2268, 0.4163, 0.4333,
 & 0.6152, 0.4906, 0.4244,
 \\
8 & $-0.6073$
 & 0.4608, 0.4816, 0.5180,
 & 0.3779, 0.3223, 0.2606,
 \\
 &
 &  0.5719, 0.6396
 & 0.1884, 0.1030
\end{tabular}
\caption{\label{paramtable}
Our calculated values of $\bar{V}_p = \min_{\vgamma,\vbeta} \lim_{n\to\infty} \EV[\braket{C/n}]$ and optimized QAOA parameters $(\vgamma,\vbeta)$.
These parameters, presented in the order of $\gamma_1,\gamma_2,\ldots$, etc., are optimized using heuristics from \cite{ZhouQAOA}.
They are confirmed to be globally optimal by running optimization from $10^4$ random starts and seeing the same values repeat.
}
\end{table}

\begin{table}[H]
\centering
\makegapedcells
\begin{tabular}{c|c|c|c}
$p$ & $V_p(\vparam)$ &  $\vgamma$ &  $-\vbeta$ \\ \hline\hline
&  & 0.2166, 0.4051, 0.4208, & 0.6226, 0.4994, 0.4410, \\
 9 & $-0.6199$ &
 0.4455, 0.4641, 0.4944,
& 0.3888, 0.3527, 0.3031,
\\
& &  0.5309, 0.5801, 0.6396
& 0.2462, 0.1769, 0.0951
\\ \hline
& & 0.2081, 0.3938, 0.4087, 0.4315,
& 0.6275, 0.5059, 0.4454, 0.4089,
\\
10 & $-0.6308$ 
& 0.4473, 0.4742, 0.5015,
& 0.3676, 0.3344, 0.2866,
\\
& & 0.5385, 0.5850, 0.6396
& 0.2321, 0.1652, 0.0878
\\ \hline
& 
& 0.2007, 0.3840, 0.3983, 0.4197, 
& 0.6317, 0.5111, 0.4525, 0.4192,  
\\
11 & $-0.6393$ 
& 0.4336, 0.4583, 0.4811, 0.5096, 
& 0.3824, 0.3444, 0.3197, 0.2730, 
\\
&
& 0.5458, 0.5895, 0.6396 
& 0.2203, 0.1553, 0.0815 
\\ \hline
& & 0.1941, 0.3753, 0.3892, 0.4096,
& 0.6353, 0.5155, 0.4583, 0.4274, 
\\
12 & $-0.6466$
& 0.4221, 0.4452, 0.4654, 0.4889,
& 0.3940, 0.3600, 0.3305, 0.3076, 
\\
&
& 0.5175, 0.5524, 0.5934, 0.6396
& 0.2615, 0.2102, 0.1468, 0.0762
%
% OLDER VERSION BELOW with more sig figs
%
%&  & 0.216559, 0.405147, 0.420751, & 0.622552, 0.499421, 0.440951, \\
% 9 & $-0.619895$ &
% 0.445457, 0.464105, 0.494433,
%& 0.388832, 0.352714, 0.303080,
%\\
%& &  0.530904, 0.580053, 0.639604
%& 0.246185, 0.176900, 0.095146
%\\ \hline
%%
%& & 0.208056, 0.393833, 0.408656, 0.431467,
%& 0.627493, 0.505860, 0.445384, 0.408897,
%\\
%10 & $-0.630782$ 
%& 0.447316, 0.474193, 0.501547,
%& 0.367610, 0.334368, 0.286583,
%\\
%& & 0.538464, 0.584989, 0.639604
%& 0.232093, 0.165188, 0.087750
%\\ \hline
%%
%& 
%& 0.200652, 0.383979, 0.398277, 0.419703, 
%& 0.631688, 0.511116, 0.452503, 0.419237,  
%\\
%11 & $-0.639312$ 
%& 0.433595, 0.458272, 0.481130, 0.509619, 
%& 0.382397, 0.344402, 0.319706, 0.272995, 
%\\
%&
%& 0.545767, 0.589450, 0.639604 
%& 0.220290, 0.155293, 0.081513 
%\\ \hline
%%
%& & 0.194123, 0.375278, 0.389216, 0.409593,
%& 0.635304, 0.515487, 0.458265, 0.427434, 
%\\
%12 & $-0.646557$
%& 0.422050, 0.445229, 0.465427, 0.488902,
%& 0.393992, 0.360004, 0.330532, 0.307588, 
%\\
%&
%& 0.517537, 0.552415, 0.593362, 0.639604
%& 0.261530, 0.210213, 0.146801, 0.076175
\end{tabular}
\caption{\label{paramtable2}
The best known values (prior to Ref.~\cite{basso2022quantum}) of $V_p=\lim_{n\to\infty} \EV[\braket{C/n}]$ for $9\le p \le 12$, and the corresponding QAOA parameters $(\vparam)$. 
These parameters are guessed by extrapolating the pattern of the optimized parameters for $p\le 8$ using interpolation and fitting. Currently, we have not tried optimizing them further since a single evaluation of $V_p(\vparam)$ for $p\ge 9$ using our procedure takes too much time.}
\end{table}

\clearpage

\section{The QAOA applied to the SK model at $p=1$\label{sec:p=1}}
For the remainder of the paper we explain how to prove our main Theorem.
We start by proving it for the simplest case of the QAOA at $p=1$ where results for any problem size $n$ can be obtained.
Since there is only one parameter in each of $\vgamma$ and $\vbeta$, we denote them without boldface.
We start by evaluating
\begin{align}
& \EVJ{\gbbraket{e^{i\lambda C/n}}} \nonumber \\
	&\qquad \qquad\qquad= \EVJ{\braket{s | e^{i\gamma C} e^{i\beta B} e^{i\lambda C/n} e^{-i\beta B} e^{-i\gamma C}|s} }
	\nonumber  \\
    &\qquad \qquad \qquad=   \EVJ{~ \sum_{\zv^1, \zv^\M, \zv^2} \braket{s|e^{i\gamma C}| \zv^1} \braket{\zv^1|e^{i\beta B}| \zv^\M} e^{i\lambda C(\zv^\M)/n} \braket{\zv^\M|e^{-i\beta B}|\zv^2} \braket{\zv^2| e^{-i\gamma C}|s}}
    	 \nonumber \\
    &\qquad \qquad\qquad = \frac{1}{2^n} \sum_{\zv^1, \zv^\M, \zv^2}  
    \EVJ{ \exp \Big({ \footnotesize i\gamma [C(\zv^1) - C(\zv^2)] + i \frac{\lambda}{n} C(\zv^\M)}\Big) }
    	 \braket{\zv^1|e^{i\beta B}| \zv^\M} \braket{\zv^\M|e^{-i\beta B}|\zv^2},
\end{align}
where we have inserted complete sets of basis vectors $\ket{\zv^1}, \ket{\zv^\M}, \ket{\zv^2}$, each of which runs through all $2^n$ elements of $\{+1,-1\}^n$.
Note that $J_{jk}$ only appears in $C(\zv)=(1/\sqrt{n})\sum_{j<k} J_{jk} z_j z_k$. 
For every string $\zv^\M$, we can redefine 
%$J_{jk} \to J_{jk} z^\M_j z^\M_k$, 
$\zv^1 \to \zv^1 \zv^\M$, $\zv^2 \to \zv^2 \zv^\M$, where the product of strings $\vect{y}\vect{z}$ is understood as bit-wise product (i.e., $[\vect{y}\vect{z}]_k = y_k z_k$).
%Since we are averaging over all choices of $J_{jk}$ via $\EVJ{\cdots}$, this is the same as averaging over all choices of $J_{jk} z^\M_j z^\M_k$. 
Therefore,
\begin{align}
&\EVJ{\gbbraket{e^{i\lambda C/n}}} \nonumber \\
&\quad~~ = \frac{1}{2^n}\sum_{\zv^1, \zv^\M, \zv^2} \EVJ{ \exp \Big({\frac{i}{\sqrt{n}}\sum_{j<k}J_{jk} z^\M_j z^\M_k \left[ \gamma(z^1_j z^1_k-z^2_j z^2_k) + \lambda/n\right]} \Big) }\braket{\zv^1\zv^\M|e^{i\beta B}|\zv^\M} \braket{\zv^\M|e^{-i\beta B}|\zv^2\zv^\M}.
\end{align}
Observe that $\braket{\zv^1\zv^\M|e^{i\beta B}|\zv^\M}=\braket{\zv^1|e^{i\beta B}|\vect{1}}$, which is independent of $\zv^\M$.
Since we have assumed $J_{jk}$ follows a symmetrical distribution, averaging over $J_{jk}$ is the same as averaging over $J_{jk} z^\M_j z^\M_k$. 
(Note this is the only place we use the assumption that the distribution of $J$ is symmetrical.) 
%an asymmetrical subgaussian distribution where $\EV[J_{jk}]=0$ and $\EV[J_{jk}^2]=1$, by expanding in the large $n$ limit to see that $\EV_J[\exp(\frac{i}{\sqrt{n}} J_{jk} z_j^\M z_k^\M x)] \approx 1-\frac{1}{2n} \EV[(J_{jk}z_j^\M z_k^\M)^2 ]  x^2 \approx 1-\frac{1}{2n} x^2$.
Hence, we can drop the dependence on $\zv^\M$, and the sum over $\zv^\M$ kills the $1/2^n$ factor to yield
\begin{align}
\EVJ{\gbbraket{e^{i\lambda C/n}}} &= \sum_{\zv^1, \zv^2} \EVJ{ \exp \Big({\frac{i}{\sqrt{n}}\sum_{j<k}J_{jk}\left[ \gamma(z^1_j z^1_k-z^2_j z^2_k) + \lambda/n\right]} \Big) }\braket{\zv^1|e^{i\beta B}|\vect{1}} 
\braket{\vect{1}|e^{-i\beta B}|\zv^2} \nonumber \\
&=  \sum_{\zv^1 \zv^2} \prod_{j<k}  \EVJ{\exp \Big({\frac{i}{\sqrt{n}}J_{jk} [\phi_{jk}(\zv^1,\zv^2) + \lambda/n]} \Big)} f_\beta(\zv^1) f_\beta^*(\zv^2),
\end{align}
where we have denoted
\begin{align}
\phi_{jk}(\zv^1,\zv^2) &\equiv \gamma (z^1_j z^1_k - z^2_j z^2_k),
\end{align}
and
\begin{align}
f_\beta(\zv) &\equiv \braket{\zv|e^{i\beta B}|\vect{1}} 
= (\cos\beta)^{\# 1\text{'s in } \zv} (i \sin\beta)^{\# -1 \text{'s in } \zv}.
\label{eq:fbeta}
\end{align}

Now, let us evaluate $\EVJ{\cdots}$, where the choices of $J_{jk}$ are averaged over.
While in principle $J_{jk}$ can be drawn from any symmetrical distribution with mean 0 and variance 1, here we consider the case of the standard normal distribution,  because it is more convenient when $J_{jk}$ is in the exponential. 
Then we have $\EVJ{e^{iJ_{jk} x}}=e^{-x^2/2}$ and so
\begin{align}
\EVJ{\gbbraket{e^{i\lambda C/n}}} &=  \sum_{\zv^1 \zv^2} \prod_{j<k}  \exp\Big[-\frac{1}{2n}(\phi_{jk} + \lambda/n)^2 \Big] f_\beta(\zv^1) f_\beta^*(\zv^2) \nonumber \\
&= \sum_{\zv^1 \zv^2} \exp\Big[-\frac{1}{2n}\sum_{j<k}\phi_{jk}^2  - \frac{\lambda}{n^2} \sum_{j<k} \phi_{jk}- \frac{\lambda^2}{2n^3} \binom{n}{2} \Big] f_\beta(\zv^1) f_\beta^*(\zv^2).
\label{eq:p1-string-basis}
\end{align}

\paragraph{Changing from string to configuration basis}---
Instead of summing over all $(2^n)^2$ possible strings $\zv^1, \zv^2$ in \eqref{eq:p1-string-basis}, we can switch to a more convenient basis that we call the \emph{configuration basis}.
For a given choice of strings $(\zv^1, \zv^2)$, if we look at the $k$-th index bit of both strings, $(z^1_k, z^2_k)$, it can only be one of four possible configurations from the following set
\begin{align}
A = \{(+1, +1), ~ (+1, -1), ~ (-1, +1), ~ (-1, -1)\}.
\end{align}
We denote the number of times the $(+1,+1)$ configuration appears among the $n$ possible bits as $n_{++}$, etc.
In other words, we want to change the basis from
\begin{equation}
\{(\zv^1, \zv^2): \zv^j \in \{\pm1\}^{2n}\}
\quad \longrightarrow \quad
 \{(n_{++}, n_{+-}, n_{-+}, n_{--}):
\sum_{\av \in A} n_{\av} = n\}.
\end{equation}
We wish to express \eqref{eq:p1-string-basis} in the configuration basis.
We start by observing that
\begin{align}
f_\beta(\zv^1) &= (\cos\beta)^{n_{++}+n_{+-}} (i\sin\beta)^{n_{-+} + n_{--}},
\end{align}
and
\begin{align}
f_\beta^*(\zv^2) &= (\cos\beta)^{n_{++}+n_{-+}} (-i\sin\beta)^{n_{+-} + n_{--}}.
\end{align}
Therefore,
\begin{align}
f_\beta(\zv^1) f_\beta^*(\zv^2)  = (\cos^2\beta)^{n_{++}} (i\sin\beta\cos\beta)^{n_{-+}} (-i\sin\beta\cos\beta)^{n_{+-}}(\sin^2\beta)^{n_{--}} \equiv \prod_{\av\in A} Q_{\av}^{n_{\av}},
\end{align}
where
\begin{align}
Q_{++} = \cos^2\beta, \qquad Q_{--} = \sin^2\beta,
\qquad \text{and} \qquad Q_{-+}= -Q_{+-} = i\sin\beta\cos\beta.
\end{align}

Now we want to perform the basis change on the part that depends on $\sum_{j<k} \phi_{jk}^q(\zv^1, \zv^2)$, where $q=1,2$.
Since $\phi_{jk}=\phi_{kj}$ and $\phi_{kk}=0$, we can write
\begin{align}
\sum_{j<k} \phi_{jk}^q = \frac12 \sum_{j,k=1}^n \phi_{jk}^q = \frac12 \sum_{\av,\bv \in A} \Phi_{\av\bv}^q n_\av n_\bv
\end{align}
where the product of configurations $\av\bv$ is understood as bit-wise product (i.e., $[\av\bv]_j = a_j b_j$), and
\begin{align}
\Phi_{\av\bv} \equiv \gamma(a_1b_1 - a_2b_2).
\end{align}
We can now rewrite \eqref{eq:p1-string-basis} in the configuration basis, yielding
\begin{equation}
\EVJ{\gbbraket{e^{i\lambda C/n}}} =  
	e^{\large -\frac{\lambda^2(n-1)}{4n^2}}
	\sum_{\{n_\av\}} \binom{n}{\{n_\av\}}  
	 \exp\Big[{-\frac{1}{4n}\sum_{\av,\bv\in A}\Phi_{\av\bv}^2 n_\av n_\bv - \frac{\lambda}{2n^2} \sum_{\av,\bv\in A}\Phi_{\av\bv}n_\av n_\bv} \Big]
	 \prod_{\av\in A} Q_\av^{n_{\av}}
	 \label{eq:expC}
\end{equation}
where the multinomial coefficient is defined in general as
\begin{align} \label{eq:multinomial}
\binom{n}{\{n_\av\}}
= \binom{n}{n_1, n_2, n_3, \ldots } =  \frac{n!}{n_1 ! n_2! n_3! \cdots }
\qquad \text{subject to} \qquad
n_1 + n_2 + n_3 + \cdots = n.
\end{align}

\paragraph{Expressions of moments}---
The desired moments $\EVJ{\gbbraket{C/n}}$ and $\EVJ{\gbbraket{(C/n)^2}}$ can be calculated by differentiating with respect to $\lambda$ and then setting $\lambda=0$.
For the first moment, we get
\begin{align}
M_1 &= \EVJ{\gbbraket{C/n}} = -i\frac{\partial}{\partial\lambda}\EVJ{\gbbraket{e^{i\lambda C/n}}}\Big|_{\lambda = 0} \nonumber \\
&=	\sum_{\{n_\av\}} \binom{n}{\{n_\av\}}
	\left(\frac{i}{2n^2} \sum_{\av,\bv\in A}\Phi_{\av\bv}n_\av n_\bv\right)
	 \exp\Big({ -\frac{1}{4n}\sum_{\av,\bv\in A}\Phi_{\av\bv}^2 n_\av n_\bv}\Big)
	 \prod_{\av\in A} Q_\av^{n_{\av}}.
	 \label{eq:Cexp-p1}
\end{align}
Similarly, the second moment is
\begin{align}
M_2 &= \EVJ{\gbbraket{(C/n)^2}} = (-i)^2\frac{\partial^2}{\partial\lambda^2} \EV_J[\braket{e^{i\lambda C/n}}]\Big|_{\lambda=0} \nonumber \\
&= \sum_{\{n_{\av}\}} \binom{n}{\{n_{\av}\}} 
	\bigg[\Big(\frac{i}{2n^2}\sum_{\av,\bv\in A}\Phi_{\av\bv}n_\av n_\bv\Big)^2 + \frac{n-1}{2n^2}\bigg]
	 \exp\Big({ -\frac{1}{4n}\sum_{\av,\bv\in A}\Phi_{\av\bv}^2 n_\av n_\bv}\Big)
	 \prod_{\av\in A} Q_{\av}^{n_{\av}} \nonumber \\
&= \frac{n-1}{2n^2} + \sum_{\{n_{\av}\}} \binom{n}{\{n_{\av}\}} 
\Big(\frac{i}{2n^2}\sum_{\av,\bv\in A}\Phi_{\av\bv}n_\av n_\bv\Big)^2
	\exp\Big({ -\frac{1}{4n}\sum_{\av,\bv\in A}\Phi_{\av\bv}^2 n_\av n_\bv}\Big)
	\prod_{\av\in A} Q_{\av}^{n_{\av}},
	 \label{eq:Csq-exp-p1}
\end{align}
where we pulled out $(n-1)/(2n^2)$  out of the sum as the rest of the summand 
corresponds to setting $\lambda=0$ on the LHS of \eqref{eq:expC}, which yields $\EV_J[\braket{\Id}]=1$.
To simplify the expressions of the moments further, note that 
\begin{align}
\sum_{\av,\bv \in A} \Phi_{\av\bv}n_\av n_\bv = \sum_{\av,\bv\in A} \gamma(a_1b_1 - a_2b_2) n_\av n_\bv\,.
\end{align}
Now the only non-zero coefficient in front of $n_\av n_\bv$ is $2\gamma a_1b_1$ when $a_1b_1=-a_2b_2$.
Thus
\begin{align}
\sum_{\av,\bv \in A} \Phi_{\av\bv}n_\av n_\bv = 4\gamma(n_{++}-n_{--})(n_{+-} - n_{-+}).
\label{eq:Phi-explicit-p1}
\end{align}
Furthermore, since $\Phi_{\av\bv}^2$ can only be 0 or $(2\gamma)^2$, we have
\begin{align}
\sum_{\av,\bv\in A}\Phi_{\av\bv}^2 n_\av n_\bv
	&= 2(2\gamma)^2(n_{++}+n_{--})(n_{+-} + n_{-+}),
\label{eq:Phisq-explicit-p1}
\end{align}
where the extra factor of 2 comes from double counting in the sum.

\paragraph{Evaluating the first moment $\braket{C/n}$}---
Plugging \eqref{eq:Phi-explicit-p1} and \eqref{eq:Phisq-explicit-p1} into \eqref{eq:Cexp-p1}, we get that the first moment is
\begin{align}
M_1 
	&= \frac{i2\gamma}{n^2}\sum_{\{n_\av\}} \binom{n}{\{n_\av\}}
	(n_{++}-n_{--})(n_{+-} - n_{-+})
	 \exp \Big[{-\frac{2\gamma^2}{n}(n_{++}+n_{--})(n_{+-} + n_{-+})} \Big]
	 \prod_{\av\in A} Q_\av^{n_{\av}}.
\end{align}
To evaluate this, let $t=n_{-+}+n_{+-}$ and $n - t=n_{++} + n_{--}$.
We then sum on $t$ to get
\begin{align}
M_1 = \frac{i2\gamma}{n^2} \sum_{t=0}^n \binom{n}{t}
	\exp\Big[{-2\gamma^2 \frac{t (n - t)}{n} }\Big]
	\sum_{n_{+-}+n_{-+}=t} &\binom{t}{n_{+-}, n_{-+}} (n_{+-}- n_{-+}) Q_{+-}^{n_{+-}} Q_{-+}^{n_{-+}}
	\nonumber \\
	 \times \sum_{n_{++}+n_{--}=n-t} &\binom{n - t}{n_{++}, n_{--}} 	(n_{++} - n_{--})Q_{++}^{n_{++} } Q_{--}^{n_{--}}.
	 \label{eq:p1-moment1-sum}
\end{align}
Observe the following identity
\begin{align}\label{eq:binom-sum-identity}
\sum_{p+q=s} \binom{s}{p,q} (p-q) x^p y^q = s (x-y)(x+y)^{s-1}.
\end{align}
We can then evaluate 
 the sum over $n_{++}$ and $n_{--}$ to obtain
\begin{align}
\sum_{n_{++}+n_{--}=n-t} \binom{n - t}{n_{++}, n_{--}} 	(n_{++} - n_{--})Q_{++}^{n_{++} } Q_{--}^{n_{--}}
= (n-t)(\cos^2\beta - \sin^2 \beta) = (n-t)\cos 2\beta,
\end{align}
where we used the fact that $Q_{++} = \cos^2\beta$ and $Q_{--} = \sin^2\beta$.

Now for the sum over $n_{+-}$ and $n_{-+}$, we have that $Q_{+-} = - Q_{-+}= - i\sin\beta\cos\beta$.
Thus when applying \eqref{eq:binom-sum-identity}, where $x=-y$, it is only nonzero for $s=1$, in which case it is $2x$.
For the $n_{+-}$ and $n_{-+}$ sum we set $t=1$ and get
\begin{align}
\sum_{n_{+-}+n_{-+}=t} \binom{t}{n_{+-}, n_{-+}} (n_{+-}- n_{-+}) Q_{+-}^{n_{+-}} Q_{-+}^{n_{-+}}
= 2 (-i \sin\beta\cos\beta) = -i \sin2\beta.
\end{align}
Now, returning to \eqref{eq:p1-moment1-sum}, we get with $t=1$
\begin{align}
\EVJ{\gbbraket{C/n}} &= M_1 
	= \frac{i2\gamma}{n^2} n
	 \exp \big({-2\gamma^2 \frac{n-1}{n}} \big)   [-i(n-1)\sin2\beta \cos 2\beta] \nonumber \\
	 &=  \frac{n-1}{n}\gamma  \exp\big({-2\gamma^2 \frac{n-1}{n}} \big)\sin 4\beta.
	 \label{eq:C-p1-n}
\end{align}
Here at $p=1$, we have obtained a formula for any finite $n$ when $J_{jk}$ are drawn from the standard normal distribution.
We can carry out a similar calculation for when $J_{jk}$ are uniformly drawn from $\{+1,-1\}$, and obtain a slightly different answer.
Nevertheless, in the infinite size limit, both answers agree and become the following
\begin{align} \label{eq:C-p1}
\boxed{
V_1(\gamma,\beta) \equiv \lim_{n\to\infty} \EVJ{\gbbraket{C/n}} = \gamma e^{-2\gamma^2 }\sin4\beta. 
}
\end{align}
This is minimized at $\gamma=1/2$ and $\beta=-\pi/8$, where
\begin{align}
V_1(1/2,-\pi/8) = -1/\sqrt{4e} \approx -0.303265.
\end{align}

\paragraph{Evaluating the second moment $\braket{(C/n)^2}$} ---
To evaluate the second moment, we plug \eqref{eq:Phi-explicit-p1} and \eqref{eq:Phisq-explicit-p1} into \eqref{eq:Csq-exp-p1} to obtain
\begin{align}
M_2 &= \frac{n-1}{2n^2} 
- \frac{4\gamma^2}{n^4}\sum_{\{n_\av\}} \binom{n}{\{n_\av\}}
	(n_{++}-n_{--})^2(n_{+-} - n_{-+})^2
	 \exp\Big[{ -\frac{2\gamma^2}{n}(n_{++}+n_{--})(n_{+-} + n_{-+})}\Big] 
	 \prod_{\av\in A} Q_\av^{n_{\av}}.
\end{align}
Note that the sum we need to do is very similar to that for the first moment \eqref{eq:p1-moment1-sum}.
We perform the same change of variable with $t=n_{-+}+n_{+-}$ and $n-t = n_{++}+n_{--}$ to get
\begin{align}
M_2 = \frac{n-1}{2n^2} - \frac{4\gamma^2}{n^4} \sum_{t=0}^n \binom{n}{t}
	e^{-2\gamma^2 t (n-t)/n }
	\sum_{n_{+-}+n_{-+}=t} 
		&\binom{t}{n_{+-}, n_{-+}} (n_{+-}- n_{-+})^2 Q_{+-}^{n_{+-}} Q_{-+}^{n_{-+}}
	\nonumber \\
	 \times \sum_{n_{++}+n_{--}=n-t} 
	 	&\binom{n - t}{n_{++}, n_{--}} (n_{++} - n_{--}) ^2Q_{++}^{n_{++} } Q_{--}^{n_{--}}.
	 \label{eq:p1-moment2-sum}
\end{align}
Using another version of the identity in \eqref{eq:binom-sum-identity},
\begin{align}
\sum_{p+q=s}\binom{s}{p, q} (p-q)^2 x^p y^q 
	=  s \big[s (x-y)^2 + 4xy \big] (x+y)^{s-2},
	\label{eq:binom-formula-2}
\end{align}
we can then evaluate the sum over $n_{++}$ and $n_{--}$ to obtain
\begin{align}
\sum_{n_{++}+n_{--}=n-t} 
	 	\binom{n - t}{n_{++}, n_{--}} (n_{++} - n_{--}) ^2Q_{++}^{n_{++} } Q_{--}^{n_{--}}
= \frac12 (n-t)[n-t+1 + (n-t-1)\cos4\beta],
\end{align}
where we plugged in $Q_{++} = \cos^2\beta$ and $Q_{--} = \sin^2\beta$.

Now for the sum over $n_{+-}$ and $n_{-+}$, we again have $Q_{+-} = - Q_{-+}= - i\sin\beta\cos\beta$.
And so in \eqref{eq:binom-formula-2} we have $x=-y$, and it is only nonzero for $s=2$, which which case it is $8x^2$.
Then
\begin{align}
\sum_{n_{+-}+n_{-+}=t} 
	\binom{t}{n_{+-}, n_{-+}} (n_{+-}- n_{-+})^2 Q_{+-}^{n_{+-}} Q_{-+}^{n_{-+}}
= 8(-i\sin\beta\cos\beta)^2 = -2\sin^2 2\beta.
\end{align}
Returning to our expression in \eqref{eq:p1-moment2-sum}, we set $t=2$ and get
\begin{align} \label{eq:M2-p=1}
M_2 
	&= \frac{n-1}{2n^2} + \frac{2\gamma^2(n-1)(n-2) [n-1+ (n-3)\cos 4\beta]}{n^3}  e^{-4\gamma^2 (n-2)/n }
	\sin^2 2\beta .
\end{align}
In the infinite size limit, this has a much simpler form
\begin{align}
\boxed{
\lim_{n\to\infty}\EVJ{\gbbraket{(C/n)^2}} = \gamma^2 e^{-4\gamma^2}\sin^2 4\beta = V_1^2(\gamma,\beta) .
}
\label{eq:Csq-p1}
\end{align}
This is exactly equal to the first moment squared in the infinite size limit.
As discussed in Section~\ref{sec:results}, this implies concentration over instances and measurements.

\paragraph{Finite-Size Effect}---
To get a sense of the effect from finite $n$, we subtract the square of \eqref{eq:C-p1-n} from \eqref{eq:M2-p=1} to get
\begin{align}
&\EVJ{\gbbraket{(C/n)^2}} - \EV_J^2[\gbbraket{C/n}]  \nonumber \\
&\qquad \qquad = \frac{1}{n}
\Big(\frac12 + 2\gamma^2 e^{-4\gamma^2}\big[\gamma^2 - (\gamma^2-1)\cos8\beta - \cos4\beta\big]
\Big) + O\Big(\frac{1}{n^2}\Big) .
\end{align}
This implies that the combined variance over instances and measurements is of order $1/n$.
So the standard deviation over each is $O(1/\sqrt{n})$.

We also note that we could find optimal parameters for finite $n$ by optimizing $\gamma$ and $\beta$ in \eqref{eq:C-p1-n}.
We see that the finite-size optimal $\gamma$ at $p=1$ differ from the infinite size limit by $O(1/n)$, while $\beta$ remains the same.

\paragraph{Generic QAOA states}---
Looking at \eqref{eq:C-p1} one concludes that for any values of $\gamma$ and $\beta$ the quantum expectation of the cost function is of order $n$.
However, as we now explain, in a generic QAOA state the expected value of the cost function is exponentially small in $n$.
To understand this apparent inconsistency, return to \eqref{eq:SK-model} where you see the factor of $1/\sqrt{n}$ in front.
Without this factor and with the $J_{jk}$'s chosen from $\{+1, -1\}$,  $\gamma$ is an angle between $-\pi$ and $\pi$.
With this convention the expected value of the cost function scales like the right hand side of \eqref{eq:C-p1} with $\gamma$ replaced by  $\sqrt{n}\gamma$.  Then, a random $\gamma$ between $-\pi$ and $\pi$ gives an exponentially small quantum expected value. 
Under this convention the expected value of the cost function is only of order $n$ in a vanishingly small subset of parameter space (i.e., the corner where $\gamma$ is of order $1/\sqrt{n}$), and exponentially small in $n$ elsewhere. 
This observation is consistent with the ``barren plateau'' prediction~\cite{barrenplateau2018} arguing that under certain conditions, the parameter search in variational quantum algorithms may be challenged by the exponentially small gradients seen at generic parameters.
%While this would ordinarily imply these algorithms are hard to train with random initialized parameters, o
What we have shown is that this issue can be addressed by avoiding random initialization in our parameter search, since we can know in advance in which corner to look.

\section{General $p$\label{sec:general-p}}
In this section we go beyond the $p=1$ calculations done in Section~\ref{sec:p=1} to prove our main Theorem for arbitrary $p$.
We will be able to evaluate the first and second moment of $C/n$ in an arbitrary QAOA state averaged over instances in the $n\to\infty$ limit, for any given $2p$ parameters.
The form of our answer will involve an iterative procedure that can be run on a classical computer with $O(16^p)$ complexity, corresponding to the formula given in Section~\ref{sec:results}.

%The form of our answer will involve an iterative procedure that needs to be run on a computer. We have implemented this routine, and although the computational complexity is $O(16^p)$ we have been able to take this out to $p=12$ and then find optimal parameters up to $p=8$. We are also able to show that in this limit,
%\begin{align}
%\lim_{n\to\infty} \EV_J[\bgbbraket{(C/n)^2}] = \lim_{n\to\infty} \EV_J^2[\bgbbraket{C/n}]
%\end{align}
%so the concentration results of Section~\ref{sec:concentration} do apply for $p >1$. 

To begin, we look at
\begin{align}
\bgbbraket{e^{i\lambda C/n}} &= \braket{s | e^{i\gamma_1 C} e^{i\beta_1 B}\cdots e^{i\gamma_p C} e^{i\beta_p B} e^{i\lambda C/n} e^{-i\beta_p B} e^{-i\gamma_p C}\cdots e^{-i\beta_1 B} e^{-i\gamma_1 C}|s}
\end{align}
and now insert $2p+1$ complete sets to get
\begin{align}
\braket{e^{i\lambda C/n}}  &= \sum_{\zv^\P{\pm1}, \ldots, \zv^\P{\pm p}, \zv^\M }
	\braket{s |\zv^\P{1}} e^{i\gamma_1 C(\zv^\P{1})} 
	\braket{\zv^\P{1}| e^{i\beta_1 B}|\zv^\P{2}} 
		\cdots e^{i\gamma_p C(\zv^\P{p})} \braket{\zv^\P{p}| e^{i\beta_p B} |\zv^\M}
		e^{i\lambda C(\zv^\M)/n} \nonumber \\
	&  \qquad \qquad \times \braket{\zv^\M |e^{-i\beta_p B} |\zv^\P{-p}} e^{-i\gamma_p C(\zv^\P{-p})}
	 \cdots \braket{\zv^\P{-2}| e^{-i\beta_1 B}|\zv^\P{-1}}  e^{-i\gamma_1 C(\zv^\P{-1})} \braket{\zv^\P{-1}|s}.
\end{align}
Here we label the $2p+1$ strings as $\zv^\P{1}, \zv^\P{2}, \ldots, \zv^\P{p}, \zv^\M, \zv^\P{-p},\ldots, \zv^\P{-2}, \zv^\P{-1}$.
This labelling is convenient because $\zv^\P{j}$ will often be paired with $\zv^\P{-j}$ in the calculations that follow.
Note each of the terms of the form say $\braket{\zv^\P{2} | e^{i\beta_3 B} |\zv^\P{3}}$ only depends on the bit-wise product $\zv^\P{2} \zv^\P{3}$.
Hence, we define
\begin{gather}
f_j(\zv \zv') \equiv \braket{\zv|e^{i\beta_j B}|\zv'}.
\end{gather}
Then we get
\begin{align} \label{eq:last-zm-sum}
\bgbbraket{e^{i\lambda C/n}} 
=  \frac{1}{2^n} \sum_{\zv^\P{\pm1}, \ldots, \zv^\P{\pm p}, \zv^\M }
	&\exp\Big[{i \sum_{r=1}^p \gamma_r [C(\zv^\P{r})-C(\zv^\P{-r})] + i\frac{\lambda}{n} C(\zv^\M)}\Big] \nonumber \\
&   \times f_1(\zv^\P{1} \zv^\P{2}) \cdots f_{p-1}(\zv^\P{p-1} \zv^\P{p}) f_{p}(\zv^\P{p} \zv^\M)
 \nonumber \\
&   \times f_{1}^*(\zv^\P{-1} \zv^\P{-2}) \cdots f_{p-1}^*(\zv^\P{1-p} \zv^\P{-p} ) f_{p}^*(\zv^\P{-p}\zv^\M ).
\end{align}
We are going to transform the $\zv^\P{j}$'s to achieve two results.
We will make each of the $f$ factors depends on only one $\zv^\P{j}$.
We will also be able to absorb $\zv^\M$ into other variables so it disappears.
To this end, for every $r=1,2,\ldots, p$, we simultaneously transform
\begin{align}
\begin{split}
\zv^\P{r} &\to \zv^\P{r} \zv^\P{r+1} \cdots \zv^\P{p} \zv^\M, \\
\zv^\P{-r} &\to  \zv^\P{-r} \zv^\P{-r-1} \cdots \zv^\P{-p} \zv^\M 
\end{split}
\end{align}
where the product of strings is understood as bit-wise product.
In the $f$ factors, we see for example
\begin{align}
\begin{split}
\zv^\P{2}\zv^\P{3} &\to (\zv^\P{2}\zv^\P{3} \cdots \zv^\P{p}\zv^\M)(\zv^\P{3}\zv^\P{4} \cdots \zv^\P{p}\zv^\M) = \zv^\P{2},\\
\zv^\P{-3}\zv^\P{-4}  &\to (\zv^\P{-3}\zv^\P{-4}\cdots \zv^\P{-p}\zv^\M)(\zv^\P{-4}\zv^\P{-5}\cdots \zv^\P{-p}\zv^\M) = \zv^\P{-3}.
\end{split}
\end{align}
In general, the products $\zv^\P{\pm r}\zv^\P{\pm(r+1)}\to \zv^\P{\pm r}$ for $1\le r \le p-1$, so the first goal is achieved.
Since $\zv^\P{\pm p} \zv^\M \to \zv^\P{\pm p}$, the second goal is also achieved.
Recall that
\begin{align}
C(\zv) = \frac{1}{\sqrt{n}} \sum_{j<k} J_{jk} z_j z_k \,.
\end{align}
Then after the transformation, $J_{jk}$ always appear together with $z^\M_j z^\M_k$ as $J_{jk} z^\M_j z^\M_k$.
Now, note that averaging over $J_{jk}$ is the same as averaging over $J_{jk} z^\M_j z^\M_k$, so we get
\begin{align}
\EVJ{\bgbbraket{e^{i\lambda C/n}} } &=
	\sum_{\zv^\P{\pm1}, \ldots, \zv^\P{\pm p}}
	\EV_J\Big[ \exp\big({\frac{i}{\sqrt{n}} \sum_{j<k} J_{jk} (\phi_{jk}  + \frac{\lambda}{n}) }\big)\Big]
		 \nonumber \\
& \quad \quad \times f_1(\zv^\P{1}) f_2(\zv^\P{2}) \cdots f_{p}(\zv^\P{p}) f_{p}^*( \zv^\P{-p}) \cdots f_{2}^*(\zv^\P{-2} ) f_{1}^*(\zv^\P{-1})
\end{align}
where
\begin{align}
\phi_{jk} &= 
	\gamma_1 \left(z^\P{1}_j z^\P{2}_j \cdots z^\P{p}_j    z^\P{1}_k z^\P{2}_k \cdots z^\P{p}_k 
	~ - ~ z^\P{-1}_j z^\P{-2}_j  \cdots z^\P{-p}_j  z^\P{-1}_k z^\P{-2}_k \cdots z^\P{-p}_k \right) + \nonumber \\
& \quad~  \gamma_2 \left(z^\P{2}_j  \cdots z^\P{p}_j z^\P{2}_k  \cdots z^\P{p}_k 
	~ - ~ z^\P{-2}_j \cdots z^\P{-p}_j  z^\P{-2}_k \cdots z^\P{-p}_k \right) + \cdots + \nonumber \\
& \quad~ \gamma_p \left(z^\P{p}_j z^\P{p}_k - z^\P{-p}_j z^\P{-p}_k\right).
\end{align}
The sum over $\zv^\M$, which no longer appears, killed the $1/2^n$ factor in front of \eqref{eq:last-zm-sum}.

For convenience, we assume each $J_{jk}$ comes from the standard normal distribution, which means $\EV_J[e^{i J_{jk} x}] = e^{-x^2/2}$.
We now average over all $J_{jk}$ to get
\begin{align}
\EVJ{\bgbbraket{e^{i\lambda C/n}} } &=
	\sum_{\zv^\P{\pm1}, \ldots, \zv^\P{\pm p}}
	\prod_{j<k} \exp\Big[-\frac{1}{2n}(\phi_{jk}+\frac{\lambda}{n})^2\Big]
	\nonumber \\
& \quad \quad \times f_1(\zv^\P{1}) f_2(\zv^\P{2}) \cdots f_{p}(\zv^\P{p}) f_{p}^*( \zv^\P{-p}) \cdots f_{2}^*(\zv^\P{-2} ) f_{1}^*(\zv^\P{-1}).
\label{eq:general-p-string-basis}
\end{align}

%% ======== CHANGING BASIS =====
\paragraph{Changing from string to configuration basis}---
Instead of summing over all $(2^n)^{2p}$ possible string values of $(\zv^\P{1},\ldots, \zv^\P{p}, \zv^\P{-p}, \ldots, \zv^\P{-1})$ in \eqref{eq:general-p-string-basis}, we will sum over a configuration basis as we did at $p=1$.
To define this configuration basis, let
\begin{align}
A =  \{+1, -1\}^{2p} =  \{\av = (a_1, a_2, \ldots, a_p, a_{-p}, \ldots, a_{-2}, a_{-1}) : a_{\pm j} \in \{+1,-1\} \},
\end{align}
where we index the $2p$ bits in each $\av$ as $1,2,\ldots, p, -p, \ldots, -2, -1$.
Now look at the $k$-th indexed bit of all the strings, that is
\begin{equation}
(z^\P{1}_k, z^\P{2}_k \ldots, z^\P{p}_k, z^\P{-p}_k, \ldots, z^\P{-2}_k, z^\P{-1}_k) \in A.
\end{equation}
We denote the number of times a given configuration $\av\in A$ occurs among the $n$ possible bits as $n_\av$, where
\begin{align}
\sum_{\av\in A} n_\av = n.
\end{align}
Note that $\av$ takes $2^{2p}$ possible values, and
\begin{align}
\sum_{\{n_{\av}\}} \binom{n}{\{n_\av\}} = \sum_{n_1, \ldots, n_{2^{2p}}} \binom{n}{n_1,\ldots, n_{2^{2p}}} = (2^{2p})^n = (2^n)^{2p},
\end{align}
where the multinomial coefficient is defined in \eqref{eq:multinomial}, so we have everything covered in the configuration basis.

We want to express the sum in \eqref{eq:general-p-string-basis}  in terms of the $\av$'s and the $n_\av$'s.
Let us start with the $f$'s, for example
\begin{align}
f_3(\zv^\P{3}) = \braket{\zv^\P{3}| e^{i\beta_3 B}|\vect{1}} 
= (\cos\beta_3)^{\# \text{ of 1's in } \zv^\P{3}} (i\sin\beta_3)^{\# \text{ of $-1$'s in } \zv^\P{3}}.
\end{align}
Let $\av = (a_1, \ldots, a_p, a_{-p}, \ldots, a_{-1})$. Then the $\cos\beta_3$ term fires when $a_3=1$ while the $i\sin\beta_3$ term fires when $a_3 = -1$.
So this $f_3$ term can be written as
\begin{align}
f_3 = (\cos\beta_3)^{\sum_\av n_\av (1+a_3)/2} (i\sin\beta_3)^{\sum_\av n_\av (1-a_3)/2}.
\end{align}
Therefore, the product of all the $f$'s is
\begin{align}
f_1(\zv^\P{1}) f_2(\zv^\P{2}) \cdots f_{p}(\zv^\P{p}) f_{p}^*( \zv^\P{-p}) \cdots f_{2}^*(\zv^\P{-2} ) f_{1}^*(\zv^\P{-1})
= \prod_{\av\in A} Q_\av^{n_\av}
\end{align}
where
\begin{align} \label{eq:Qdef}
Q_\av = \prod_{j=1}^p
	(\cos\beta_j)^{1+(a_j + a_{-j})/2}
	(\sin\beta_j)^{1-(a_j + a_{-j})/2}
	(i)^{(a_{-j} - a_j)/2} \,.
\end{align}
	
Return to \eqref{eq:general-p-string-basis} where we want to put the $\phi_{jk}$ terms into the configuration basis.
Note for $q=1,2$,
\begin{align}
\sum_{j<k} \phi_{jk}^q = \frac12 \sum_{j,k=1}^n \phi_{jk}^q
\end{align}
since $\phi_{jk}=\phi_{kj}$ and $\phi_{kk}=0$.
For any configuration $\cv\in A$, let
\begin{align} \label{eq:Phidef}
\Phi_{\cv} &= \sum_{r=1}^p \gamma_r (c_r^* - c_{-r}^*),
\end{align}
where $\cv^*$ is a transformation of $\cv$ defined as
\begin{align} \label{eq:star-op}
c^*_r=c_r c_{r+1}\cdots c_p
\qquad \text{and} \qquad
c^*_{-r} = c_{-r} c_{-r-1}\cdots c_{-p}
\quad
\text{for } 1\le r\le p \,.
\end{align}
For example, for $p=3$,
\begin{equation}
\begin{array}{lclllrrr}
\cv &= & (c_1, &c_2, &c_3, & c_{-3}, &c_{-2}, &c_{-1}),\\
\cv^* &= & (c_1c_2c_3,  &c_2 c_3, &c_3, &c_{-3} ,  & c_{-2}c_{-3} , & c_{-1} c_{-2}c_{-3}).
\end{array}
\end{equation}
Then
\begin{align} 
\sum_{j<k} \phi_{jk}^q  = \frac{1}{2}\sum_{\av,\bv \in A} \Phi_{\av\bv}^q n_\av n_\bv
\end{align}
where $\av\bv$ is understood as bit-wise product (i.e., $[\av\bv]_k =a_k b_k$).
Thus, \eqref{eq:general-p-string-basis} in the configuration basis is
\begin{align}
\EVJ{\bgbbraket{e^{i\lambda C/n}} }
	= \sum_{\{n_\av\}} \binom{n}{\{n_\av\}} &\exp\Big[{-\frac{1}{4n}\sum_{\av,\bv\in A} \Phi_{\av\bv}^2 n_\av n_\bv - \frac{\lambda}{2n^2}\sum_{\av,\bv\in A} \Phi_{\av\bv}n_\av n_\bv - \frac{\lambda^2}{2n^3}\binom{n}{2} }\Big] \prod_{\av\in A} Q_\av^{n_\av}.
	\label{eq:general-p-config}
\end{align}
which is the same as \eqref{eq:expC} in the $p=1$ case, except with expanded meaning of $\Phi_{\av\bv}$, $Q_\av$, etc.

\subsection{Organizing the sum by showing $1=1$}
\label{sec:1=1}
We are going to expand \eqref{eq:general-p-config} in $\lambda$ to get the first two moments as the coefficients of $\lambda^1$ and $\lambda^2$.
To organize our calculation we first look at the $\lambda^0$ term, which yields 1 on the left hand side of \eqref{eq:general-p-config}.
This means
\begin{align} \label{eq:1=1}
1 = \sum_{\{n_\av\}} \binom{n}{\{n_\av\}} &\exp\Big[-\frac{1}{4n}\sum_{\av,\bv\in A} \Phi_{\av\bv}^2 n_\av n_\bv \Big]
\prod_{\av\in A} Q_\av^{n_\av}.
\end{align}
However, it is not obvious how the sum on right hand side evaluates to 1. Understanding why this is true will help us evaluate the sum for the higher moments.

To make it work, we organize the sum over $\{n_\av\}$ to get cancellations.
To this end, we divide the set of configurations $A$ into $p+1$ subsets, 
\begin{equation}
A=A_1 \cup A_2 \cup \cdots \cup A_{p+1},
\end{equation}
where
\begin{align} \label{eq:Apart-def}
A_\ell = \{\av:  a_{-k}  = a_{k} \text{ for } p-\ell+1 <  k \le p, \text{ and } a_{-p+\ell- 1} = - a_{p-\ell+1}\}
\quad
\text{for}
\quad
1\le \ell \le p,
\end{align}
and
\begin{equation}
A_{p+1} = \{\av : a_{-k} = a_{k}  \text{ for } 1\le k \le p\} \,. 
\end{equation}
Recall that we index the entries of $\av$ as $(a_1,\ldots, a_p, a_{-p}, \ldots, a_{-1})$.
Subset $A_1$ has $a_{-p} = -a_{p}$, so there are $2^{2p-1}$ elements.
Subset $A_2$ has $a_{-p} = a_{p}$ and $a_{-p+1} = -a_{p-1}$, with $2^{2p-2}$ elements, etc.
We illustrate this for
$p=3$,
\begin{equation} \label{eq:A-part-example}
\begin{array}{lclllrrrr}
A_1 &: \quad     & (a_1, &a_2, &a_3, & -a_{3}, &a_{-2}, &a_{-1})   &\quad 32\text{ elements} \\
A_2 &: \quad     & (a_1, &a_2, &a_3,  & a_{3}, &-a_{2}, &a_{-1})   & \quad 16\text{ elements} \\
A_{p} &: \quad   & (a_1, &a_2, &a_3, & a_{3}, &a_{2}, &-a_{1})   & \quad 8\text{ elements} \\
A_{p+1} &: \quad & (a_1, &a_2, &a_3, & a_{3}, &a_{2}, &a_{1})   &\quad 8\text{ elements}
\end{array}
\end{equation}

We now define a ``bar'' operation that takes configuration $\av \in A_\ell$ to $\bar\av \in A_\ell$ for $1 \le \ell \le p$ via
\begin{align} \label{eq:bar-op}
\bar{a}_{\pm r} = \begin{cases}
a_{\pm r}, & r \neq p-\ell+1 \\
-a_{\pm r} & r = p-\ell+1
\end{cases}
\quad
\text{ for } \quad \av \in A_\ell, \quad 1\le \ell \le p\,.
\end{align}
Note the bar operation is its own inverse. We also never bar elements of $A_{p+1}$.
For example, for the form of $\av$'s given in \eqref{eq:A-part-example}, the corresponding $\bar\av$'s are
\begin{equation}
\begin{array}{lclrrrrrrr}
A_1 &: \quad     & (\hspace{8pt} a_1, &a_2, &-a_3, & a_{3}, &a_{-2}, &a_{-1})  & \phantom{\quad 32\text{ elements}} \\
A_2 &: \quad     & (\hspace{8pt} a_1, &-a_2, &a_3,  & a_{3}, &a_{2}, &a_{-1})  &\\
A_{p} &: \quad   & (-a_1, &a_2, &a_3, & a_{3}, &a_{2}, &a_{1})   & 
\end{array}
\end{equation}
Note $\bar{\av}$ flips the signs of two elements of $\av$.
If you look at $Q_\av$ defined in \eqref{eq:Qdef}, we see that
\begin{align}
Q_{\bar\av} = -Q_\av \,.
\end{align}
Furthermore, we note that
\begin{align}
\Phi_{\av\bar\av} = 0\,.
\end{align}

We now introduce a notion we call ``plays well with'':
\begin{defn} \label{def:playwell}
Given two configurations $\av,\bv\in A$, we say $\av$ \emph{plays well with} $\bv$ if $\Phi_{\av\bv}^2 = \Phi_{\bar\av\bv}^2.$
\end{defn}

\begin{lemma} \label{lem:ordering}
For any $p$, if $\av\in A_\ell$ and $\bv \in A_{\ell'}$, where $\ell\le \ell'$, then $\av$ plays well with $\bv$.
\end{lemma}
\begin{proof}
We show how to prove for the case of $p=3$, and the generalization to other $p$ is immediate.
First, let us show that all $\av \in A_1$ plays well with all $\bv$. Using $p=3$ as an example, then $\av$ and $\bv$ are of the form
\begin{equation}
\begin{array}{ccccccccc}
\av &= & (a_1, &a_2, &a_3,   &-a_{3}, &a_{-2}, &a_{-1}), \\
\bv &= & (b_1,  &b_2, &b_3,  &b_{-3},  &b_{-2}, &b_{-1}).
\end{array}
\end{equation}
Then $\Phi_{\av\bv}$ depends on the following
\begin{align}
(\av\bv)^* = (a_1 a_2 a_3 b_1 b_2 b_3, \quad
			  a_2 a_3 b_2 b_3, \quad
			  a_3 b_3, \quad
			  -a_3 b_{-3}, \quad
			  -a_{-2} a_3 b_{-2} b_{-3}, \quad
			  -a_{-1} a_{-2} a_3   b_{-1} b_{-2} b_{-3}).
\end{align}
Note every term is proportional to $a_3$. Taking $\av$ to $\bar\av$ flips the sign of $a_3$,
so $\Phi_{\bar\av\bv} = -\Phi_{\av\bv}$ and $\Phi_{\bar\av\bv}^2 = \Phi_{\av\bv}^2$.
We see that indeed all $\av\in A_1$ plays well with all $\bv\in A$.

Now consider $\av\in A_2$ and $\bv\in A_3$, where
\begin{equation}
\begin{array}{ccccccccc}
\av &= & (a_1, &a_2, &a_3,   &a_{3}, &-a_{2}, &a_{-1}), \\
\bv &= & (b_1,  &b_2, &b_3,  &b_{3},  &b_{2}, &-b_{1}).
\end{array}
\end{equation}
So
\begin{align}
(\av\bv)^* = (a_1 a_2 a_3 b_1 b_2 b_3, \quad
			  a_2 a_3 b_2b_3, \quad
			  a_3 b_3, \quad
			  a_3 b_{3}, \quad
			  -a_{2} a_3 b_{2} b_{3}, \quad
			  a_{-1} a_{2} a_3  b_{1} b_{2} b_{3}).
\end{align}
Now the coefficients of $\gamma_3$ in $\Phi_{\av\bv}$ is 0, and all other terms flip sign when $a_2\to -a_2$, so again we have  $\Phi_{\bar\av\bv}^2 = \Phi_{\av\bv}^2$ when $\av\in A_2$ and $\bv\in A_3$.
The other cases and generalization to other $p$ are immediate.
\end{proof}

Return to our task of showing $1=1$ as in \eqref{eq:1=1}.
Pick one $\dv\in A_1$ and let $A^c = A\setminus\{\dv,\bar{\dv}\}$.
We can write the RHS of \eqref{eq:1=1} as
\begin{align}
\sum_{t_d=0}^n
~ \sum_{n_\dv+ n_{\bar\dv}=t_d}
  \sum_{\{n_\av:~ \av\in A^c\}} 
&\binom{n}{t_d} \binom{t_d}{n_\dv, n_{\bar\dv}} \binom{n-t_d}{\{n_{\av}\}}
\exp\Big[{-\frac{1}{4n}\sum_{\av,\bv\in A^c} \Phi_{\av\bv}^2 n_\av n_\bv}\Big]\prod_{\av\in A^c} Q_\av^{n_\av} \nonumber \\
&  \exp\Big[{-\frac{1}{2n}\sum_{\bv\in A^c} \Phi^2_{\dv\bv} n_\dv n_\bv}\Big]
\exp\Big[{-\frac{1}{2n}\sum_{\bv\in A^c} \Phi^2_{\bar\dv\bv} n_{\bar\dv} n_\bv}\Big]
Q_{\dv}^{n_{\dv}} Q_{\bar\dv}^{n_{\bar\dv}}.
\label{eq:1=1-d-example}
\end{align}
Now note the identity
\begin{align} \label{eq:binom-ident}
\sum_{n_{\dv} + n_{\bar\dv}=t_d} \binom{t_d}{n_{\dv}, n_{\bar\dv}} x^{n_{\dv}} y^{n_{\bar\dv}} = (x+y)^{t_d}.
\end{align}
If $y=-x$ this is only non-zero (and equal to 1) if $t_d=0$, in which case $n_\dv = n_{\bar\dv}=0$.
And since $\dv\in A_1$ plays well with all $\bv \in A^c$ and $Q_{\bar\dv} = -Q_\dv$, we have that $y=-x$, so in fact $n_\dv = n_{\bar\dv}=0$ in \eqref{eq:1=1-d-example}.

This argument applies to all $\dv \in A_1$. So we have $n_{\dv}=0$ for all $\dv\in A_1$. Now take $\dv\in A_2$. This $\dv$ plays well with all the remaining $\bv$'s which are at levels 2 through $p+1$.
So now all $\dv\in A_2$ have $n_{\dv}=0$.
Continuing in this fashion, the RHS of \eqref{eq:1=1} becomes
\begin{align}
\sum_{\{n_\av:~\av\in A_{p+1}\}} \binom{n}{\{n_\av\}}\exp\Big[-\frac{1}{4n}\sum_{\av,\bv\in A_{p+1}} \Phi_{\av\bv}^2 n_\av n_\bv\Big]\prod_{\av\in A_{p+1}} Q_\av^{n_\av}.
\end{align}
But consider $\Phi_{\av\bv}$ with any $\av,\bv\in A_{p+1}$, 
\begin{equation}
\begin{array}{lclllrrl}
\av &= & (a_1, &a_2, &a_3,   &a_{3}, &a_{2}, &a_{1}), \\
\bv &= & (b_1,  &b_2, &b_3,  &b_{3},  &b_{2}, &b_{1}),
\end{array}
\end{equation}
so $(\av\bv)_r^* = (\av\bv)_{-r}^*$, which gives $\Phi_{\av\bv}=0$ as can be seen from its definition from \eqref{eq:Phidef}.
So now the RHS of \eqref{eq:1=1} becomes
\begin{align}
\sum_{\{n_\av:~\av\in A_{p+1}\}} \binom{n}{\{n_\av\}}
	\prod_{\av\in A_{p+1}} Q_\av^{n_\av}  = \bigg[\sum_{\av\in A_{p+1}} Q_{\av}\bigg]^n.
\end{align}
Return to the form of $Q_\av$ given in \eqref{eq:Qdef}.
For $\av\in A_{p+1}$ where $a_{-j}= a_j$ for all $1\le j \le p$, we have
\begin{align}
Q_\av = \prod_{j=1}^p (\cos^2\beta_j)^{(1+a_j)/2} (\sin^2\beta_j)^{(1-a_j)/2}
\end{align}
and
\begin{align} \label{eq:Qsum}
\sum_{\av\in A_{p+1}} Q_{\av} &= \sum_{a_1=\pm1} \cdots \sum_{a_p=\pm1} \prod_{j=1}^p (\cos^2\beta_j)^{(1+a_j)/2} (\sin^2\beta_j)^{(1-a_j)/2} = \prod_{j=1}^p [\cos^2\beta_j + \sin^2\beta_j] = 1,
\end{align}
so happily we have 1 = 1.

\subsection{Evaluating the moments}
We are now going to calculate the first and second moments, and we will organize the sum in the same way as for the zeroth moment, but evaluating the sum will be more complicated.
From \eqref{eq:general-p-config}, the first moment is
\begin{align} \label{eq:M1}
\EVJ{\bgbbraket{C/n}}  
	= \frac{i}{2n^2}\sum_{\{n_\av\}} \binom{n}{\{n_\av\}}
	\bigg(\sum_{\uv,\vv\in A} \Phi_{\uv\vv}n_\uv n_\vv\bigg)
	 \exp\Big[{-\frac{1}{4n}\sum_{\av,\bv\in A} \Phi_{\av\bv}^2 n_\av n_\bv }\Big]   \prod_{\av\in A} Q_\av^{n_\av}.
\end{align}
Similarly, the second moment is
\begin{align} \label{eq:M2}
&\EVJ{\bgbbraket{(C/n)^2}}  \nonumber \\
	&\qquad= \frac{n-1}{2n^2} - \frac{1}{4n^4}\sum_{\{n_\av\}} \binom{n}{\{n_\av\}}
	\bigg(\sum_{\uv,\vv\in A} \Phi_{\uv\vv}n_\uv n_\vv\bigg)^2
	 \exp\Big[{-\frac{1}{4n}\sum_{\av,\bv\in A} \Phi_{\av\bv}^2 n_\av n_\bv }\Big]   \prod_{\av\in A} Q_\av^{n_\av}.
\end{align}
In the limit of $n\to\infty$, the first term drops out and we only care about the second term.
To evaluate these two moments, we first note that
\begin{align} 
\Phi_{\uv\vv} = \sum_{r=1}^p \gamma_r \left[(\uv^*\vv^*)_r - (\uv^*\vv^*)_{-r} \right].
\end{align}
Recall the $*$ operation defined in \eqref{eq:star-op}.
Note each term is non-zero only when $u^*_r u^*_{-r} = -v^*_r v^*_{-r}$. We then write a conveniently factorized form
\begin{align}\label{eq:factorized}
\sum_{\uv,\vv\in A} \Phi_{\uv\vv} n_\uv n_\vv = \sum_{r=1}^p \sum_{\uv,\vv\in A} \gamma_r (u^*_r + u^*_{-r}) (v^*_r-v^*_{-r}) n_\uv n_\vv.
\end{align}
Therefore, the hard part of the calculation is to evaluate sums of the form
\begin{align} \label{eq:Suv-def}
\Suv =  \sum_{\{n_\av\}} \binom{n}{\{n_\av\}}
	 \exp\Big[{-\frac{1}{4n}\sum_{\av,\bv\in A} \Phi_{\av\bv}^2 n_\av n_\bv }\Big]  \Big(\prod_{\av\in A} Q_\av^{n_\av} \Big) \frac{n_\uv n_\vv}{n^2} \\
\Suvxy =  \sum_{\{n_\av\}} \binom{n}{\{n_\av\}}
	 \exp\Big[{-\frac{1}{4n}\sum_{\av,\bv\in A} \Phi_{\av\bv}^2 n_\av n_\bv }\Big]  \Big(\prod_{\av\in A} Q_\av^{n_\av} \Big) \frac{n_\uv n_\vv n_\xv n_\yv}{n^4}
	 \label{eq:Suvxy-def}
\end{align}
for a fixed $\uv, \vv, \xv, \yv$.
To that end,  we will show the following Lemma:
\begin{lemma}\label{lem:W}
There are a set of functions $\{W_\uv(\vparam): \uv\in A\}$ such that
\begin{align} \label{eq:Suv}
\lim_{n\to\infty} 
	 \Suv = W_\uv W_\vv
\end{align}
and
\begin{align} \label{eq:Suvxy}
\lim_{n\to\infty} 
	 \Suvxy = W_\uv W_\vv W_\xv W_\yv \,.
\end{align}
\end{lemma}

\noindent
Combining \eqref{eq:M1} with \eqref{eq:factorized} and  \eqref{eq:Suv}, we have
\begin{align} \label{eq:Vp}
V_p(\vparam) = \lim_{n\to\infty} \EVJ{\bgbbraket{C/n}}  = \frac{i}{2} \sum_{r=1}^p \sum_{\uv,\vv\in A} \gamma_r (u^*_r + u^*_{-r}) (v^*_r-v^*_{-r}) W_\uv W_\vv
\end{align}
which is the expected value of the energy of the QAOA applied to typical instance of SK model in the infinite size limit,
in terms of the yet-to-be-determined $W$'s that we will soon show how to compute.
Combining \eqref{eq:M2} with  \eqref{eq:factorized} and \eqref{eq:Suvxy} gives
\begin{align} \label{eq:concentration}
&\lim_{n\to\infty} \EVJ{\bgbbraket{(C/n)^2}} \nonumber \\
&\qquad \qquad= -\frac{1}{4} \sum_{r,s=1}^p \sum_{\uv,\vv,\xv,\yv\in A} \gamma_r \gamma_s (u^*_r + u^*_{-r}) (v^*_r-v^*_{-r}) (x^*_s + x^*_{-s}) (y^*_s-y^*_{-s}) W_\uv W_\vv W_\xv W_\yv \nonumber \\
&\qquad \qquad= [V_p(\vparam)]^2 . % = \lim_{n\to\infty} \EV_J^2[\bgbbraket{C/n}].
\end{align}
This proves our main Theorem in Section~\ref{sec:results}.
%Applying the arguments in Section~\ref{sec:concentration} shows both concentration over $J$ and concentration over measurements for fixed $J$ for the case of general $p$.

%The reader who is not interested in more details can skip to Section~\ref{sec:implementation} where \eqref{eq:Vp} is repeated as \eqref{eq:Vp-box}. There we present a self-contained procedure for evaluating \eqref{eq:Vp} for general $p$.

\paragraph{Proof of Lemma~\ref{lem:W} -- First Moment}---
To begin we show \eqref{eq:Suv}.
Note in evaluating $\Suv$ for the purpose of calculating the first moment, we see from \eqref{eq:factorized} that we need not consider the case where $\uv=\vv$ or $\uv=\bar{\vv}$.
Even though \eqref{eq:Suv} can be shown to be true for all choice of $\uv,\vv$, here we will only prove it for the cases we need. The other cases are simple extensions of the arguments presented here.
Let us denote
\begin{align} \label{eq:gdef}		
g_{\av\bv} = \exp\Big[{-\frac{1}{2n} \Phi^2_{\av\bv} }\Big]
\end{align}
to keep some of the formulas a bit more compact.
We write \eqref{eq:Suv-def} as
\begin{align} \label{eq:Suv-w-g}
\Suv = \sum_{\{n_\av\}} \binom{n}{\{n_\av\}}
	 \prod_{\av,\bv\in A}g_{\av\bv}^{ n_\av n_\bv/2 }  \prod_{\av\in A} Q_\av^{n_\av} \frac{n_\uv n_\vv}{n^2}~.
\end{align}

We will organize \eqref{eq:Suv-w-g} by summing over $\{n_\av: \av \in A_{p+1}\}$ first.
For these $\av$'s, the $Q_{\av}$'s have trigonometric factors but no $i$'s. When showing that $1=1$ these were the only non-zero $n_\av$'s, so they added to $n$.
The remaining configurations are in the set we denote as
\begin{equation}
B \equiv A\setminus A_{p+1}
\end{equation}
which is everything in $A$ that is not in $A_{p+1}$.
Now let
\begin{align}
\sum_{\av\in A_{p+1}} n_\av = n-t
\qquad
\text{and}
\qquad
\sum_{\bv\in B} n_\bv = t.
\end{align}
We can write
\begin{align} \label{eq:Suv-suv}
\Suv = \sum_{t=0}^n \suv(t,n)
\end{align}
where
\begin{align}  \label{eq:Suv-with-I}
\suv(t,n) = \binom{n}{t} 
&\sum_{\{n_\bv:~ \bv\in \Ac\}}\binom{t}{\{n_\bv\}} 
\prod_{\bv, \bv'\in \Ac} g_{\bv\bv'}^{ n_\bv n_{\bv'}/2} \prod_{\bv \in \Ac} Q_\bv^{n_\bv} \nonumber \\
\times
& \underbrace{\sum_{\{n_\av:~ \av\in A_{p+1}\}} \binom{n-t}{\{n_\av\}}
\prod_{\av\in A_{p+1},\bv\in \Ac} g_{\av\bv}^{ n_\av n_\bv} 
	 \prod_{\av \in A_{p+1}} Q_\av^{n_\av}  }_{\I}  ~ \frac{n_\uv n_\vv}{n^2}
\end{align}
where we used that $\Phi^2_{\av\av'}=0$ if $\av,\av'$ are both in $A_{p+1}$.

To evaluate $\Suv$ and control our approximation as $n\to\infty$, our strategy is to examine the properties of its summand $\suv$.
We will first show that for fixed $T$, independent of $n$, and $t\le T$, 
\begin{align} \label{eq:luv-approx}
\suv(t,n) = \luv(t) \times (1+O(1/n))
\end{align}
for some $\luv$ that depends on $t$ and not $n$.
This will lead to our goal which is \eqref{eq:Suv=luv}.
To this end, we need to show that there exists a sequence $\{\buv(t)\}_{t=0}^\infty$, which depends on $t$ but not $n$ that bounds $|\suv(t,n)|$ via
\begin{align} \label{eq:buvdef}
|\suv(t,n)| \le \buv(t), \quad \text{for all} \quad t\le n,
\end{align}
and satisfies 
\begin{align} \label{eq:buv-converge} 
\sum_{t=0}^\infty \buv(t) <\infty
\,.
\end{align}
Once we have this sequence $\buv$ satisfying \eqref{eq:buvdef} and \eqref{eq:buv-converge}, then
% of \eqref{eq:buvdef}
\begin{align} \label{eq:Suv=luv-ineq}
\bigg|\Suv - \sum_{t=0}^T \luv(t)\bigg| \le \sum_{t=0}^T \bigg|\suv(t,n) - \luv(t)\bigg| + \sum_{t=T+1}^n \buv(t)\,.
\end{align}
By taking the limit as $n\to\infty$, we see that the first term of the RHS vanishes because of \eqref{eq:luv-approx}, yielding
\begin{align}
\lim_{n\to\infty}\bigg|\Suv - \sum_{t=0}^T \luv(t)\bigg| \le \sum_{t=T+1}^\infty \buv(t)\,.
\end{align}
 Then taking the limit as $T\to\infty$, we see from \eqref{eq:buv-converge} that
\begin{align} \label{eq:Suv=luv}
\lim_{n\to\infty} \Suv = \sum_{t=0}^\infty \luv(t)~.
\end{align}
The last remaining step is to show that the RHS of \eqref{eq:Suv=luv} evaluates to $W_\uv W_\vv$ as asserted in \eqref{eq:Suv} in Lemma~\ref{lem:W}.

We will proceed in two cases: (1) neither $\uv$ nor $\vv$ is in $A_{p+1}$; or (2) at least one of $\uv,\vv$ is in $A_{p+1}$. 
We first focus on finding $\luv(t)$ in \eqref{eq:luv-approx} which we need for \eqref{eq:Suv=luv}. 
We will later show \eqref{eq:buvdef} and \eqref{eq:buv-converge} which are needed to prove \eqref{eq:Suv=luv} is true.

\paragraph{Case 1:} Suppose $\uv, \vv$ are both not in $A_{p+1}$.
Using the multinomial generalization of \eqref{eq:binom-ident}, we can sum over $\{n_\av: \av \in A_{p+1}\}$ in \eqref{eq:Suv-with-I} with $t$ fixed to get
\begin{align}  \label{eq:Idef}
\I = \left[\sum_{\av \in A_{p+1}} Q_\av \prod_{\bv\in \Ac} g_{\av\bv}^{n_\bv}\right]^{n-t}.
\end{align}
Since $\bv \in \Ac$, we have that $n_\bv \le t$.
For $t\le T$ with $T$ fixed, and for large $n$, we can expand $g_{\av\bv}$ in \eqref{eq:gdef} to leading order in $1/n$, giving
\begin{align} \label{eq:I}
\I &= \left[\sum_{\av \in A_{p+1}} Q_\av \Big({1- \frac{1}{2n}\sum_{\bv\in \Ac} \Phi_{\av\bv}^2 n_\bv} \Big) + O(\frac{1}{n^2})\right]^{n-t}.
\end{align}
Using the fact that $\sum_{\av\in A_{p+1}} Q_\av = 1$ from \eqref{eq:Qsum}, we have
\begin{align}
\I &=\left[1 - \frac{1}{2n}\sum_{\av \in A_{p+1}} Q_\av\sum_{\bv\in \Ac} \Phi_{\av\bv}^2 n_\bv + O(\frac{1}{n^2}) \right]^{n-t} ~.
\end{align}
For large $n$  we have
\begin{align}
\I &= \exp\left[ - \frac12\sum_{\av \in A_{p+1}} Q_\av \sum_{\bv\in \Ac} \Phi_{\av\bv}^2 n_\bv\right]\times (1 + O(1/n)) ~.
\end{align}
Let us define
\begin{align} \label{eq:Fdef}
F_\bv = \exp\left[-\frac{1}{2}\sum_{\av \in A_{p+1}} Q_\av \Phi_{\av\bv}^2\right].
\end{align}
then
\begin{align}
\I &= \prod_{\bv\in \Ac} F_\bv^{n_\bv} \times \big(1 + O(1/n)\big) ~.
\label{eq:I-form}
\end{align}
Note that since any $\bv \in \Ac = A$\,$\setminus$\,$A_{p+1}$ plays well with any $\av\in A_{p+1}$, we have that
\begin{equation}
F_\bv = F_{\bar\bv} ~.
\end{equation}
Returning to \eqref{eq:Suv-with-I}, we have for $t\le T$
\begin{align}
\label{eq:Suv-simple}
\suv(t,n) = \binom{n}{t}  \sum_{\{n_\bv:~ \bv\in \Ac\}}
\binom{t}{\{n_\bv\}}
 \prod_{\bv, \bv'\in \Ac} g_{\bv\bv'}^{n_\bv n_{\bv'}/2}
 	  \prod_{\bv \in \Ac} X_\bv^{n_\bv}   \Big(\frac{n_\uv n_\vv}{n^2}\Big) \big(1 + O (1/n)\big)
\end{align}
where we defined
\begin{align} \label{eq:Xdef}
X_\bv = Q_\bv F_\bv.
\end{align}
Note $X_{\bar\bv} = -X_\bv$ for $\bv\in \Ac$. Evaluating \eqref{eq:Suv-simple} will allow us to obtain an approximate form of $\suv(t,n)$ as $\luv(t)$.

In order to  establish the patterns  that make this calculation doable, let
us pretend that $\Ac$ has 8 elements and write
\begin{equation}
\Ac = \{1, \bar1, 2, \bar2, 3, \bar3, 4, \bar4\}.
\end{equation}
Because of Lemma~\ref{lem:ordering}, we can assume without loss of generality that the lower numbered elements play well with the higher (e.g., $2$ and $\bar2$ play well with $3,\bar3, 4, \bar4$ but not necessarily with $1,\bar1$, etc.).
We are not assuming anything about the levels these elements come from other than that they are consistent with the playing-well ordering.
Now, between 1 and 2 we have
\begin{align}
g_{12}^{n_1 n_2} g_{\bar1 2}^{n_{\bar1} n_2} g_{1\bar2}^{n_1 n_{\bar2}} g_{\bar1 \bar2}^{n_{\bar1} n_{\bar2}}.
\end{align}
But 1 plays well with both 2 and $\bar2$, so $g_{12} = g_{\bar1 2}$ and $g_{1\bar2} = g_{\bar1 \bar2}$, so this factor is
\begin{align}
(g_{12}^{n_2} g_{1\bar2}^{n_{\bar2}})^{n_1 + n_{\bar1}}.
\end{align}
Let us also write
\begin{align}
t_j = n_j + n_{\bar j} ~.
\end{align}
Then \eqref{eq:Suv-simple} becomes
\begin{align} \label{eq:Suv-example}
\suv(t,n) = \sum_{t_1+t_2+t_3+t_4=t}\binom{n}{t} \binom{t}{t_1, t_2, t_3, t_4}
 \sum_{\{n_j + n_{\bar{j}} = t_j\}}
&\binom{t_1}{n_1, n_{\bar1}} \gf{1}{2} \gf{1}{3} \gf{1}{4} \Xf{1} \nonumber\\
\times&\binom{t_2}{n_2, n_{\bar2}} \gf{2}{3} \gf{2}{4} \Xf{2} \nonumber\\
\times&\binom{t_3}{n_3, n_{\bar3}} \gf{3}{4} \Xf{3} \nonumber \\
\times&\binom{t_4}{n_4, n_{\bar4}} \Xf{4} \nonumber\\
\times& \Big(\frac{n_\uv}{n}\Big)\Big( \frac{n_\vv}{n}\Big)\times \Big(1+ O\Big(\frac{1}{n}\Big)\Big).
\end{align}

Now, for given $t_j$, we will sum over $n_j$ and $n_{\bar j}$.
If $j$ or $\bar{j}$ are not $\uv$ or $\vv$, these sums are of the form
\begin{align} \label{eq:binom-0down}
\sum_{n_j + n_{\bar j} = t_j} \binom{t_j}{n_j n_{\bar{j}}} D_j^{n_j} D_{\bar j}^{n_{\bar j}} = (D_j + D_{\bar j})^{t_j}.
\end{align}
But we also need to consider the cases where $j$ or $\bar{j}$ can be $\uv$ or $\vv$.
As mentioned earlier, we don't need to consider cases where $\uv=\vv$ or $\uv=\bar\vv$ when evaluating $\Suv$.
Then we have sums of the form
\begin{align} 
\label{eq:binom-1down}
\sum_{n_j + n_{\bar j} = t_j} n_j \binom{t_j}{n_j n_{\bar{j}}} 
D_j^{n_j} D_{\bar j}^{n_{\bar j}} &= t_j D_j (D_j + D_{\bar j})^{t_j-1},
\end{align}
and
\begin{align}
\sum_{n_j + n_{\bar j} = t_j} n_{\bar{j}} \binom{t_j}{n_j n_{\bar{j}}} D_j^{n_j} D_{\bar j}^{n_{\bar j}} &= t_j D_{\bar{j}} (D_j + D_{\bar j})^{t_j-1}.
\end{align}

\subparagraph{Example summing on $n_3, n_{\bar3}$ where $\uv,\vv\neq 3,\bar3$}---
As an example, let us do the sum on $n_3$ and $n_{\bar3}$ in \eqref{eq:Suv-example}, first assuming that $3$ and $\bar{3}$ are not equal to $\uv$ or $\vv$.
Using \eqref{eq:binom-0down}, this sum gives
\begin{align} \label{eq:Xg3-example}
\Big(X_3 \gt{1}{3} \gt{2}{3} + X_{\bar3} \gt{1}{\bar3} \gt{2}{\bar3}\Big)^{t_3}.
\end{align}
Now, for $t\le T$ with $T$ fixed and $n$ large, we have
\begin{align}
X_3 \gt{1}{3} \gt{2}{3} + X_{\bar3} \gt{1}{\bar3} \gt{2}{\bar3}
&=  X_3\exp\left[-\frac{1}{2n} (\Phi^2_{13} t_1 + \Phi^2_{23} t_2)\right]
	- X_3 \exp\left[-\frac{1}{2n} (\Phi^2_{1\bar3} t_1 + \Phi^2_{2\bar3} t_2)\right] \nonumber \\
& = \frac{1}{2n}\left[X_3(\Phi^2_{\bar31} - \Phi^2_{31}) t_1 + X_3 (\Phi^2_{\bar32} - \Phi^2_{32}) t_2\right] + O(1/n^2) \nonumber \\
& = \frac{1}{n}\left(K_{3,1} t_1 + K_{3,2} t_2\right) + O(1/n^2)
\label{eq:K3example}
\end{align}
where we defined
\begin{align} \label{eq:Kdef}
K_{\bv,\cv} = \frac12 X_\bv (\Phi^2_{\bar\bv \cv} - \Phi^2_{\bv \cv}) ~.
\end{align}
Now the combinatorial factor in \eqref{eq:Suv-example} is
\begin{align}
\binom{n}{t} \binom{t}{t_1, t_2, t_3, t_4}  = \frac{n!}{(n-t)!} \frac{1}{t_1! t_2!t_3!t_4!} ~.
\end{align}
For $t\le T$ with $T$ fixed and $n$ large, we have $n!/(n-t)!=n^t (1+O(1/n))$, and so
\begin{align}
\binom{n}{t} \binom{t}{t_1, t_2, t_3, t_4}
= \frac{n^{t_1}}{t_1!} \frac{n^{t_2}}{t_2!} \frac{n^{t_3}}{t_3!} \frac{n^{t_4}}{t_4!} \times \big(1 + O(1/n)\big).
\end{align}
The factor of $n^{t_3}/t_3!$ combines with \eqref{eq:Xg3-example} and \eqref{eq:K3example} to give 
\begin{align} \label{eq:Kt-example}
\frac{1}{t_3!} (K_{3,1} t_1 + K_{3,2} t_2)^{t_3} \times (1+O(1/n)).
\end{align}

So far, we have summed on $n_3$ and $n_{\bar3}$ with $t_3=n_3 + n_{\bar3}$, but the formula given in \eqref{eq:Kt-example} has a 3 but no $\bar{3}$.
This is because
\begin{align} \label{eq:K-identity1}
K_{\bar{\bv},\cv} = K_{\bv,\cv}
\end{align}
as can be seen as follows: $K_{\bv,\cv} =\frac12 X_{\bv} (\Phi^2_{\bar\bv,\cv} - \Phi^2_{\bv,\cv})$ as defined in \eqref{eq:Kdef}, but $X_{\bar\bv} = - X_\bv$, so $K_{\bar\bv,\cv}=K_{\bv,\cv}$. We also have that
\begin{align} \label{eq:K-identity2}
K_{\bv,\bar\cv} = K_{\bv, \cv} ~.
\end{align}
To see this, note that $K_{\bv,\cv}$ is not zero only if $\bv$ does not play well with $\cv$, and thus $\bv$ is after $\cv$ in the play-well ordering by Lemma~\ref{lem:ordering}. Then $\cv$ must play well with $\bv$ and $\bar\bv$, and thus $\Phi_{\bv\bar\cv}^2 = \Phi_{\bv\cv}^2$ and $\Phi_{\bar\bv\bar\cv}^2 = \Phi_{\bar\bv\cv}^2$, implying \eqref{eq:K-identity2}. 
So in \eqref{eq:Kt-example}, we could replace $1$ with $\bar{1}$, $2$ with $\bar{2}$, or $3$ with $\bar3$.

\subparagraph{Example summing over $n_3, n_{\bar3}$ where $\uv=3$}---
Now, let us see what happens if say $\uv=3$. Including $n_\uv / n$ in the sum over $n_3$ and $n_{\bar3}$ as well as $n^{t_3}/t_3!$ from the combinatorial factor, we can use \eqref{eq:binom-1down} and get
\begin{align} \label{eq:sum-n3-u3}
\frac{1}{n}\frac{n^{t_3}}{t_3!} t_3X_3 \gt{1}{3} \gt{2}{3} \Big(X_3 \gt{1}{3} \gt{2}{3} + X_{\bar3} \gt{1}{\bar3} \gt{2}{\bar3}\Big)^{t_3-1}.
\end{align}
For $t\le T$ with $T$ fixed and $n$ large, \eqref{eq:sum-n3-u3} becomes
\begin{align}
\frac{1}{t_3!} t_3 X_3 \gt{1}{3} \gt{2}{3}  (K_{3,1} t_1 + K_{3,2} t_2)^{t_3-1} \times (1+O(1/n)).
\end{align}
Note the factors of $n$ have all cancelled, because we have $n^{t_3}$ from the combinatorial factor, $(1/n)^{t_3-1}$ from \eqref{eq:K3example} raised to the $t_3-1$ power, and $1/n$ from $n_\uv/n$.
Then since $t_1,t_2 \le t$, we have $\gt{1}{3} \gt{2}{3}=1 + O(1/n)$ for large $n$, so we  have
\begin{align}
\frac{1}{t_3!} t_3  X_3  (K_{3,1} t_1 + K_{3,2} t_2)^{t_3-1}\times (1+O(1/n)).
\label{eq:before-shift}
\end{align}
The sum on $t_3$ starts at $t_3=0$, where this term is zero.
The $t_3=1$ term is $X_3$.
We can let $t_3'=t_3-1$ and sum from $t_3'=0$.
The factor \eqref{eq:before-shift} is now
\begin{align}
\frac{1}{t_3'!} X_3 (K_{3,1} t_1 + K_{3,2} t_2)^{t_3'}  \times (1 + O(1/n)).
\end{align}

\subparagraph{Obtaining $\luv$ to get \eqref{eq:Suv=luv}}---
Combining the two examples above, we can see that one can write \eqref{eq:Suv-example} as
\begin{align}
\suv(t,n) &= \luv(t)
	\times \big(1 + O(1/n)\big)
\end{align}
with
\begin{align} \label{eq:luv-first}
\luv(t) = X_\uv X_\vv  \sum_{t'_1+t'_2+t'_3+t'_4=t-2}
	\frac{1}{t'_1!} 0^{t'_1}
	\times \frac{1}{t'_2!} (K_{2,1} t_1)^{t'_2}
	\times \frac{1}{t'_3!} (K_{3,1} t_1 + K_{3,2} t_2)^{t'_3}
	\nonumber \\
\times
	\frac{1}{t'_4!} (K_{4,1} t_1 + K_{4,2} t_2 + K_{4,3} t_3)^{t'_4}
\end{align}
where $t_j = t'_j + 1$ if $j=\uv$ or $\vv$, and $t_j=t'_j$ otherwise.
The factor of $0^{t'_1}$ forces $t'_1$ to be zero, but we keep it in this form to help see the pattern.
We can replace $t_j$ with $t_j'$ by defining
\begin{align}\label{eq:Kbar-def}
\bar{K}_j = \sum_{i<j} K_{j,i} (\delta_{i,\uv} + \delta_{i,\vv}) = K_{j,\uv} + K_{j,\vv},
\end{align}
where the last equality comes from the fact that $K_{j,i} = 0$ if $i \ge j$.
So for example,
\begin{align}
	\frac{1}{t'_3!} (K_{3,1} t_1 + K_{3,2} t_2)^{t'_3}
	=
	\frac{1}{t'_3!} (K_{3,1} t'_1 + K_{3,2} t'_2 + \bar{K}_3)^{t'_3}.
\end{align}
In particular, we have $\bar{K}_1=0$,  since $K_{1,i}=0$ for any $i$.
Having written everything in terms of $t'_j$, we can now drop the prime to get
\begin{align} \label{eq:luv-explicit}
\luv(t) = X_\uv X_\vv \sum_{t_1+t_2+t_3+t_4=t-2}
	\frac{1}{t_1!} \bar{K}_1^{t_1}
	\times
	\frac{1}{t_2!} (K_{2,1} t_1 + \bar{K}_2 )^{t_2}
	\times
	\frac{1}{t_3!} (K_{3,1} t_1 + K_{3,2} t_2 + \bar{K}_3)^{t_3} \nonumber \\
	\times
	\frac{1}{t_4!} (K_{4,1} t_1 + K_{4,2} t_2 + K_{4,3} t_3 + \bar{K}_4)^{t_4} .
\end{align}
And we see that $\luv$ is indeed independent of $n$.

Assuming \eqref{eq:Suv=luv} is true and letting $\luv(0)=\luv(1)=0$ (for the case where $\uv,\vv\not\in A_{p+1}$), we have 
\begin{align} \label{eq:Suv-inf}
\lim_{n\to\infty} \Suv = \sum_{t=0}^\infty \luv(t) = X_\uv X_\vv \sum_{t_1,t_2,t_3,t_4=0}^\infty
	\frac{1}{t_1!} \bar{K}_1^{t_1}
	\times
	\frac{1}{t_2!} (K_{2,1} t_1 + \bar{K}_2 )^{t_2}
	\times
	\frac{1}{t_3!} (K_{3,1} t_1 + K_{3,2} t_2 + \bar{K}_3)^{t_3} \nonumber \\
	\times
	\frac{1}{t_4!} (K_{4,1} t_1 + K_{4,2} t_2 + K_{4,3} t_3 + \bar{K}_4)^{t_4}.
\end{align}
Note the rearrangement of the terms in this sum is justified because the series is absolutely convergent, as can be seen by replacing the $K$'s by their norms and summing explicitly, first on $t_4$, then $t_3$, etc.

\subparagraph{Proving \eqref{eq:Suv=luv}}---
We now take an aside to establish the truth of \eqref{eq:Suv=luv}, but the reader can skip to \eqref{eq:Suv-evaluation-start} without interruption since it does not help us in evaluating $\Suv$ for the first moment.
We need to bound $|\suv(t,n)|$ by $\buv(t)$ as in \eqref{eq:buvdef}, and show that $\buv$ gives a convergent series as in \eqref{eq:buv-converge}, then from \eqref{eq:Suv=luv-ineq} we get \eqref{eq:Suv=luv}.
The first goal is to derive an $n$-independent bound $\buv(t)$ where $|\suv(t,n)| \le \buv(t)$, for which we need to return to \eqref{eq:Suv-with-I}.
Plugging \eqref{eq:Idef} into \eqref{eq:Suv-with-I}, we get
\begin{align} \label{eq:suv-exact}
\suv(t,n)
= 
\binom{n}{t}  \sum_{\{n_\bv:~ \bv\in \Ac\}}
\binom{t}{\{n_\bv\}}
 \prod_{\bv, \bv'\in \Ac} g_{\bv\bv'}^{n_\bv n_{\bv'}/2} 
 	  \prod_{\bv \in \Ac} Q_\bv^{n_\bv}  \underbrace{\Big[\sum_{\av \in A_{p+1}} Q_\av \prod_{\bv\in \Ac} g_{\av\bv}^{n_\bv}\Big]^{n-t}}_{\I} \Big(\frac{n_\uv n_\vv}{n^2}\Big)
.
\end{align}
We observe that since everyone plays well with $\av \in A_{p+1}$, we can write the factor
\begin{align}
g_{\av\bv}^{n_{\bv}} g_{\av\bar\bv}^{n_{\bar \bv}} = g_{\av\bv}^{n_{\bv}+n_{\bar \bv}} = g_{\av\bv}^{t_\bv}
\end{align}
where $t_\bv = n_\bv + n_{\bar\bv}$.
Thus, we can bipartition $\Ac = D \cup D^c$ such that if $\bv \in D$ then $\bar\bv \in D^c$, and write
\begin{align}
\I = \Big[{\sum_{\av \in A_{p+1}} Q_\av \prod_{\bv\in \Ac} g_{\av\bv}^{n_\bv}}\Big]^{n-t} = \Big[\sum_{\av \in A_{p+1}} Q_\av \prod_{\bv\in D} g_{\av\bv}^{t_\bv}\Big]^{n-t}
\end{align}
which is a function of $\{t_\bv\}$ and not $\{n_\bv\}$. We also note that since $0\le g_{\av\bv} \le1$, and $0 \le Q_\av \le 1$ when $\av\in A_{p+1}$, we have
\begin{align} \label{eq:I-bound}
0 \le \I \le \left[{\textstyle \sum_{\av \in A_{p+1}} Q_\av }\right]^{n-t} = 1 \,.
\end{align}
Then we can bound \eqref{eq:suv-exact} by
\begin{align}
\bigg|
\suv(t,n)
\bigg|
&= 
\bigg|
\mvcenter{\sum_{\{t_\dv: \dv\in D\}}}
\binom{n}{t}
 \binom{t}{\{t_\dv\}}  \I(\{t_\dv\}) 
\mvcenter{\sum_{ \{ n_\dv+n_{\bar\dv} = t_\dv:\, \dv\in D \}} }
\left[ \mvcenter{\prod_{\dv \in D}} \binom{t_\dv}{n_\dv, n_{\bar\dv}}
	\bigg(\mvcenter{ \prod_{\bv, \bv'\in \Ac} g_{\bv\bv'}^{n_\bv n_{\bv'}/2}} \bigg)
 	 \bigg(\mvcenter{\prod_{\bv \in \Ac} Q_\bv^{n_\bv}} \bigg)
	 \frac{n_\uv n_\vv}{n^2}
	 \right]
\bigg|
\nonumber \\
&\le 
\mvcenter{\sum_{\{t_\dv: \dv\in D\}}}
\binom{n}{t}
 \binom{t}{\{t_\dv\}}  
\left|
 \mvcenter{\sum_{ \{n_\dv+n_{\bar\dv} = t_\dv:\, \dv \in D \} } }
 \left[
 \mvcenter{\prod_{\dv \in D}} \binom{t_\dv}{n_\dv, n_{\bar\dv}}
	\bigg(\mvcenter{ \prod_{\bv, \bv'\in \Ac} g_{\bv\bv'}^{n_\bv n_{\bv'}/2}} \bigg)
 	 \bigg(\mvcenter{\prod_{\bv \in \Ac} Q_\bv^{n_\bv}} \bigg)
	  \frac{n_\uv n_\vv}{n^2}
	  \right]
\right|
 \label{eq:suv-bound}
\end{align}
where we pulled $\I$ out of the sum on $n_\bv$'s and used the fact that $|\sum_k A_k B_k| \le \sum_k |A_k| |B_k|$.

Now we want to calculate an $n$-independent bound on the RHS of \eqref{eq:suv-bound}. We will use the same strategy as before that gave us $\luv$, where we looked at an example where $B$ has 8 elements.
Then we have
\begin{align} \label{eq:buv-example}
\bigg|\suv(t,n)\bigg| \le   \mvcenter{\sum_{t_1+t_2+t_3+t_4=t}}  \binom{n}{t} \binom{t}{t_1, t_2, t_3, t_4}~ \left|\mvcenter{\sum_{\{n_j + n_{\bar{j}} = t_j\}}}\right. 
 &\binom{t_1}{n_1, n_{\bar1}} \gf{1}{2} \gf{1}{3} \gf{1}{4} \Qf{1} \nonumber\\
\times&\binom{t_2}{n_2, n_{\bar2}} \gf{2}{3} \gf{2}{4} \Qf{2} \nonumber\\
\times&\binom{t_3}{n_3, n_{\bar3}} \gf{3}{4} \Qf{3} \nonumber \\
\times&\binom{t_4}{n_4, n_{\bar4}} \Qf{4} \nonumber\\
\times& \Big(\frac{n_\uv}{n}\Big)\Big( \frac{n_\vv}{n}\Big)
\Big|.
\end{align}
Then we obtain from \eqref{eq:suv-bound} an expression that is similar to \eqref{eq:Suv-example} obtained from \eqref{eq:Suv-simple}, except every $X_\bv$ is replaced by $Q_\bv$.

We can look at the same examples summing over $n_3$ and $n_{\bar3}$ as before.
Consider first again the case where $\uv,\vv \neq 3, \bar3$.
We want to bound this sum, along with the relevant factor from the multinomial coefficient involving $t_3$.
By looking at \eqref{eq:buv-example}, and noting that $n!/(n-t)!\le n^t$, the sum over $n_3$ and $n_{\bar3}$ along with the combinatorial factor gives
\begin{align}
\frac{n^{t_3}}{t_3!} \Big(Q_3 \gt{1}{3} \gt{2}{3} + Q_{\bar3} \gt{1}{\bar3} \gt{2}{\bar3}\Big)^{t_3} .
 \label{eq:bound-sum-n3}
\end{align}
We will use the following identity, which is a simple consequence of the Mean Value Theorem:
For $a,b>0$, we have
\begin{align} \label{eq:MVident}
|e^{-b} - e^{-a}|  \le |a-b| ~.
\end{align}
Then the sum on $n_3$ and $n_{\bar 3}$ is bounded by
\begin{align}
\left|\frac{n^{t_3}}{t_3!} \Big(Q_3 \gt{1}{3} \gt{2}{3} + Q_{\bar3} \gt{1}{\bar3} \gt{2}{\bar3}\Big)^{t_3} \right|
&= \frac{|Q_3|^{t_3} n^{t_3}}{t_3!} \left|\exp\Big[{-\frac{1}{2n} (\Phi^2_{13} t_1 + \Phi^2_{23} t_2)}\Big] - \exp\Big[{-\frac{1}{2n} (\Phi^2_{1\bar3} t_1 + \Phi^2_{2\bar3} t_2)}\Big]\right|^{t_3} \nonumber \\
&\le \frac{|Q_3|^{t_3} n^{t_3}}{t_3!} \left| \frac{1}{2n} \left[(\Phi^2_{\bar31} - \Phi^2_{31}) t_1 + (\Phi^2_{\bar32} - \Phi^2_{32}) t_2\right]\right|^{t_3} \nonumber \\
&= \frac{1}{t_3!} |F_3|^{-t_3} |K_{3,1}t_1 + K_{3,2} t_2|^{t_3} 
\end{align}
since $Q_\bv = X_\bv/F_\bv$.

We now consider the second example where we sum over $n_3$ and $n_{\bar3}$, but $\uv=3$.
Including the factor of $n_\uv/n$ in the sum, we get an expression just like \eqref{eq:sum-n3-u3}, but with $X_\bv$ replaced by $Q_\bv$:
\begin{align}
\frac{1}{n}\frac{n^{t_3}}{t_3!} t_3Q_3 \gt{1}{3} \gt{2}{3} \Big(Q_3 \gt{1}{3} \gt{2}{3} + Q_{\bar3} \gt{1}{\bar3} \gt{2}{\bar3}\Big)^{t_3-1}.
\end{align}
We can then apply the identity \eqref{eq:MVident} in the example above, and use the fact that $0 < g_{\av\bv} \le 1$ to show that for general $t$, this sum over $n_3$ and $n_{\bar 3}$ is bounded:
\begin{align}
\left|\frac{1}{n}\frac{n^{t_3}}{t_3!} t_3 Q_3 \gt{1}{3} \gt{2}{3} \Big(Q_3 \gt{1}{3} \gt{2}{3} + Q_{\bar3} \gt{1}{\bar3} \gt{2}{\bar3}\Big)^{t_3-1}\right| 
\le \frac{1}{t_3!} t_3 |Q_3| |F_3|^{1-t_3} |K_{3,1} t_1 +  K_{3,2} t_2|^{t_3-1}.
\end{align}

Now we can combine the two examples above just like before in \eqref{eq:luv-first}, and use the same trick that got us $\luv$ in \eqref{eq:luv-explicit} to obtain a bound 
\begin{align} 
|\suv(t,n)| &\le |Q_\uv| |Q_\vv| \sum_{t_1+t_2+t_3+t_4=t-2}
	\frac{\Ff{1}}{t_1!} \bar{K}_1^{t_1}
	\times
	\frac{\Ff{2}}{t_2!} |K_{2,1} t_1 + \bar{K}_2 |^{t_2}
	\times
	\frac{\Ff{3}}{t_3!} |K_{3,1} t_1 + K_{3,2} t_2 + \bar{K}_3|^{t_3} \nonumber \\
	&\qquad \qquad \qquad \qquad \qquad \qquad \qquad\qquad \quad     \times
	\frac{\Ff{4}}{t_4!} |K_{4,1} t_1 + K_{4,2} t_2 + K_{4,3} t_3 + \bar{K}_4|^{t_4}.
\end{align}
Then we can define
\begin{align}\label{eq:buv-explicit}
\buv(t) &= |Q_\uv| |Q_\vv| \hspace{-5pt} \sum_{t_1+t_2+t_3+t_4=t-2}\hspace{-5pt}
	\frac{\Ff{1}}{t_1!} \bar{K}_1^{t_1}
	\times
	\frac{\Ff{2}}{t_2!} (|K_{2,1}| t_1 + |\bar{K}_2| )^{t_2}
	\times
	\frac{\Ff{3}}{t_3!} (|K_{3,1}| t_1 + |K_{3,2}| t_2 + |\bar{K}_3|)^{t_3} \nonumber \\
	&\qquad \qquad \qquad \qquad \qquad \qquad \qquad\qquad \quad     \times
	\frac{\Ff{4}}{t_4!} (|K_{4,1}| t_1 + |K_{4,2}| t_2 + |K_{4,3}| t_3 + |\bar{K}_4|)^{t_4}
\end{align}
so $|\suv(t,n)|\le \buv(t)$.

We now  show \eqref{eq:buv-converge}.
Note when summing $\buv(t)$ from $t=2$ to $\infty$, 
we in fact sum over all values of $t_j$ from $0$ to $\infty$. Then 
\begin{align}
\sum_{t=2}^\infty \buv(t) &= |Q_\uv| |Q_\vv|
 \hspace{-3pt}\sum_{t_1,t_2,t_3,t_4=0}^\infty  \hspace{-3pt}
	\frac{\Ff{1}}{t_1!} \bar{K}_1^{t_1}
	\times
	\frac{\Ff{2}}{t_2!} (|K_{2,1}| t_1 + |\bar{K}_2 |)^{t_2}
	\times
	\frac{\Ff{3}}{t_3!} (|K_{3,1}| t_1 + |K_{3,2}| t_2 + |\bar{K}_3|)^{t_3} \nonumber \\
	&\qquad \qquad \qquad \qquad \qquad \qquad \qquad    \times
	\frac{\Ff{4}}{t_4!} (|K_{4,1}| t_1 + |K_{4,2}| t_2 + |K_{4,3}| t_3 + |\bar{K}_4|)^{t_4}
\end{align}
can be summed explicitly, 
summing on $t_4$, then $t_3$, then $t_2$, then $t_1$, which shows that the series converges and justifies the rearrangement of the sum.
And hence \eqref{eq:Suv=luv} is true.
Our remaining task is thus to give an iterative procedure for explicitly evaluating $\sum_{t=0}^\infty \luv(t)$.

\subparagraph{Iterative procedure to evaluate \eqref{eq:Suv-inf}}---
Here we show how to evaluate $\lim_{n\to\infty} \Suv$ given in \eqref{eq:Suv-inf} explicitly with an iterative procedure  that can be generalized to arbitrarily many $t_j$'s.
To proceed, we first multiply the summand in \eqref{eq:Suv-inf} by $\Am_1^{t_1} \Am_2^{t_2} \Am_3^{t_3} \Am_4^{t_4} \Am $ where we initialize $\Am_1=\Am_2=\Am_3=\Am_4=\Am =1$.
This gives
\begin{align} \label{eq:Suv-evaluation-start}
\lim_{n\to\infty} \Suv  = X_\uv X_\vv \sum_{t_1,t_2,t_3,t_4=0}^\infty \frac{\bar{K}_1^{t_1}}{t_1!} \frac{ (K_{2,1} t_1 + \bar{K}_2)^{t_2}}{t_2!} \frac{ (K_{3,1} t_1 + K_{3,2} t_2 + \bar{K}_3)^{t_3}}{t_3!} \nonumber \\
	\times
	\frac{ (K_{4,1} t_1 + K_{4,2} t_2 + K_{4,3} t_3 + \bar{K}_4)^{t_4}}{t_4!} \Am_1^{t_1} \Am_2^{t_2} \Am_3^{t_3} \Am_4^{t_4} \Am.
\end{align}
Now, summing on $t_4$, we get
\begin{align}
\lim_{n\to\infty} \Suv  = X_\uv X_\vv \sum_{t_1,t_2,t_3=0}^\infty \frac{\bar{K}_1^{t_1}}{t_1!} \frac{ (K_{2,1} t_1 + \bar{K}_2)^{t_2}}{t_2!} \frac{ (K_{3,1} t_1 + K_{3,2} t_2 + \bar{K}_3)^{t_3}}{t_3!} \nonumber \\
\times~
e^{\Am_4 K_{4,1} t_1} e^{\Am_4 K_{4,2}t_2} e^{\Am_4 K_{4,3} t_3} e^{\Am_4 \bar{K}_4} \Am_1^{t_1} \Am_2^{t_2} \Am_3^{t_3} \Am,
\end{align}
Let
\begin{align}
\Am_1 \gets \Am_1 e^{\Am_4 K_{4,1}}, \qquad
\Am_2 \gets \Am_2 e^{\Am_4 K_{4,2}}, \qquad
\Am_3 \gets \Am_3 e^{\Am_4 K_{4,3}}, \qquad
\Am \gets \Am e^{\Am_4 \bar{K}_4}.
\end{align}
With the new $\Am_1$, $\Am_2$, $\Am_3$ and $\Am$ we have
\begin{align}
\lim_{n\to\infty} \Suv  = X_\uv X_\vv \sum_{t_1,t_2,t_3=0}^\infty \frac{\bar{K}_1^{t_1}}{t_1!} \frac{ (K_{2,1} t_1 + \bar{K}_2)^{t_2}}{t_2!} \frac{ (K_{3,1} t_1 + K_{3,2} t_2 + \bar{K}_3)^{t_3}}{t_3!} \Am_1^{t_1} \Am_2^{t_2} \Am_3^{t_3} \Am.
\end{align}
Now the sum on $t_3$ updates the following symbols
\begin{align}
\Am_1 \gets \Am_1 e^{\Am_3 K_{3,1}}, \qquad
\Am_2 \gets \Am_2 e^{\Am_3 K_{3,2}}, \qquad
\Am \gets \Am e^{ \Am_3 \bar{K}_3}.
\end{align}
Then summing over $t_2$ gives
\begin{align}
\Am_1 \gets \Am_1 e^{\Am_2 K_{2,1}}, \qquad
\Am \gets \Am e^{ \Am_2 \bar{K}_2}.
\end{align}
And finally, the sum over $t_2$ yields
\begin{align}
\Am \gets \Am e^{\Am_1 \bar{K}_1}.
\end{align}
In summary, the final values of $\Am_1, \Am_2, \Am_3, \Am_4$ and $\Am$ are given by the following:
\begin{equation}
\begin{split}
\Am_4 &= 1, \\
\Am_3 &= e^{\Am_4 K_{4,3}}, \\
\Am_2 &= e^{\Am_4 K_{4,2}} e^{\Am_3 K_{3,2}}, \\
\Am_1 &= e^{\Am_4 K_{4,1}} e^{\Am_3 K_{3,1}} e^{\Am_2 K_{2,1}}, \\
\Am &= e^{\Am_4 \bar{K}_4}e^{\Am_3 \bar{K}_3} e^{\Am_2 \bar{K}_2} e^{\Am_1 \bar{K}_1}.
\end{split}
\end{equation}

More generally, when the sum in \eqref{eq:Suv-inf} has $N=|B|/2$ levels, we have
\begin{gather}
\label{eq:Ydef}
\Am_{N} = 1, \qquad
\Am_j = \exp\left({\textstyle \sum_{k > j} \Am_k K_{k,j}}\right) \qquad \text{for }  N-1 \ge j \ge 1,\\
\text{and} \qquad
Y = \exp\left({\textstyle \sum_{k=1}^N \Am_j \bar{K}_j}\right).
\end{gather}
Then $\lim_{n\to\infty} \Suv = X_\uv X_\vv Y$.

In fact, we can further simplify this by showing that
\begin{align}
\Am = \Am_\uv \Am_\vv \,.
\end{align}
This can be seen via the following.
Recall that the  definition of $\bar{K}_j$ from \eqref{eq:Kbar-def} in terms of $\uv$ and $\vv$ is
\begin{align}
\bar{K}_j =  K_{j, \uv} + K_{j,\vv}\,.
\end{align}
Then we have
\begin{align}
\Am 
= \exp\left({\textstyle \sum_{j=1}^{N} \Am_j K_{j,\uv}}\right) \exp\left({\textstyle \sum_{j=1}^{N} \Am_j K_{j,\vv}}\right)
= \exp\left({\textstyle \sum_{j> \uv} \Am_j K_{j,\uv}}\right) \exp\left({\textstyle \sum_{j > \vv} \Am_j K_{j,\vv}}\right) = Y_\uv Y_\vv
\end{align}
where we used the fact that $K_{j,\uv} = 0$ if $j \le \uv$.

And thus for $\uv,\vv \not\in A_{p+1}$,
\begin{equation} \label{eq:Suv-W}
\lim_{n\to\infty} \Suv = \sum_{t=0}^\infty \luv(t) = X_\uv X_\vv Y_\uv Y_\vv = W_\uv W_\vv, \qquad \text{where} \qquad W_\uv = X_\uv Y_\uv
\end{equation}
as we claimed in \eqref{eq:Suv} of Lemma~\ref{lem:W}.

We now want to give a formula for $W_\uv$ directly instead of using $Y_\uv$.
It will be convenient to define
\begin{align} \label{eq:Delta-def}
\Delta_{\bv,\cv} = \frac12 (\Phi_{\bar\bv\cv}^2 - \Phi_{\bv\cv}^2),
\end{align}
so that  $K_{\bv,\cv} =  X_\bv \Delta_{\bv,\cv} $.
Furthermore, in what's above we are only iterating over $j$ that indexes a unique pair of configurations $j, \bar{j} \in B = A \setminus A_{p+1}$, therefore implicitly we have $Y_j = Y_{\bar j}$.
This also means 
\begin{equation}
W_{\bar{j}} = X_{\bar{j}} Y_{\bar{j}} = -X_j Y_j = - W_{j}.
\end{equation}
Now,  for each $j = |B|/2,|B|/2-1, \ldots, 2, 1$, we have
\begin{align}
W_j &= X_j Y_j = X_j \exp\left({\textstyle \sum_{k > j} Y_{k} K_{k,j}}\right)
= X_j \exp\left({\textstyle \sum_{k > j} Y_{k} X_{k} \Delta_{k,j}}\right) \nonumber \\
&= X_j \exp\left({\textstyle \sum_{k > j} W_{k} \Delta_{k,j}}\right).
\label{eq:W-formula}
\end{align}
In particular, we have $W_{|B|/2} = X_{|B|/2}$.
This gives an iterative formula for $\{W_\uv: \uv \not\in A_{p+1}\}$.

However, what we have shown so far is only a proof of Lemma~\ref{lem:W} when both $\uv, \vv \not\in A_{p+1}$.
It turns out that the formula \eqref{eq:W-formula} for $W_\uv$ easily generalizes to the other case when $\uv \in A_{p+1}$, which we now consider.

\paragraph{Case 2:} Now suppose at least one of $\uv, \vv$ is in $A_{p+1}$.
If $\uv$ is in $A_{p+1}$, then $\uv = (u_1,u_2, \ldots, u_p, u_p, \ldots, u_2, u_1)$, where $u_{-r}=u_r$ for $1\le r\le p$.
It follows that $u^*_r = u_{-r}^*$.
Thus from \eqref{eq:factorized}, we see that if $\uv \in A_{p+1}$, then $\vv$ must not be.
So we need only consider the case where $\uv \in A_{p+1}$ and $\vv \in A\setminus A_{p+1}$.
Let us for now treat the $Q_\av$'s with $\av\in A_{p+1}$ as formal variables, whose values are not yet specified by \eqref{eq:Qdef}.
Then in \eqref{eq:Suv-with-I} we can replace $n_\uv/n$ by
\begin{align} \label{eq:QI}
\frac{n_\uv}{n} \quad
\longleftrightarrow 
\quad
\frac{Q_\uv}{n} \frac{\partial}{\partial Q_\uv} \,.
\end{align}
Differentiating \eqref{eq:Idef} with respect to $Q_\uv$ and then multiplying by $Q_\uv/n$ gives
\begin{align}
\frac{Q_\uv}{n} (n-t)\prod_{\bv\in \Ac} g_{\uv\bv}^{n_\bv}
\left[\sum_{\av \in A_{p+1}} Q_\av \prod_{\bv\in \Ac} g_{\av\bv}^{n_\bv}\right]^{n-t-1} .
\end{align}
For $t\le T$ with $T$ fixed and $n$ large, this is
\begin{align}
Q_\uv \I \times \big(1+O(1/n)\big) = Q_\uv \prod_{\bv\in B} F_\bv^{n_\bv} \times \big(1+O(1/n)\big) ~.
\end{align}
So the net effect of having $n_\uv/n$ with $\uv \in A_{p+1}$ is to have a factor of $Q_\uv$ instead of $X_\uv$ in $\luv$.
Observe that for any $\uv\in A_{p+1}$, if we look at the definition of $F_\uv$ in \eqref{eq:Fdef}, we have $F_\uv=1$ since $\Phi_{\uv\av}=0$ for any $\uv,\av\in A_{p+1}$.
Furthermore, for $\uv \in A_{p+1}$, with whom everyone plays well by Lemma~\ref{lem:ordering}, we have $K_{j,\uv}=\Delta_{j,\uv} =0$ for any $j$.
Then generalizing the formula that defines $W_\uv$ in \eqref{eq:W-formula} to $\uv\in A_{p+1}$, we get $W_\uv = X_\uv = Q_\uv F_\uv = Q_\uv$.
Thus the factor of $Q_\uv$ can also be written as $W_\uv$, for consistency of notation with  \eqref{eq:Suv-W} where $\uv$ can now be any element of $A$.

Hence, in general, for any $\uv, \vv \in A$, where $\uv \neq \vv, \bar\vv$, and at most one of $\uv,\vv$ is in $A_{p+1}$, we have
\begin{align}
\lim_{n\to\infty} \Suv = \sum_{t=0}^\infty \luv(t) = W_\uv W_\vv
\end{align}
as desired in \eqref{eq:Suv} in Lemma~\ref{lem:W}.

\paragraph{Proof of Lemma~\ref{lem:W} -- Second Moment}---
We need to evaluate $\Suvxy$ defined in \eqref{eq:Suvxy-def} and show that in the limit as $n\to\infty$ it is $W_\uv W_\vv W_\xv W_\yv$ as in \eqref{eq:Suvxy}.
Recall that in the first moment calculation of $\Suv$, we only had to consider cases where $\uv\neq \vv,\bar{\vv}$, and we obtained the $W_\uv W_\vv$ result.
Now for the present calculation, if $\uv,\vv,\xv$ and $\yv$ are all distinct, the first moment argument carries through, and we immediately get $W_\uv W_\vv W_\xv W_\yv$.
Looking at \eqref{eq:factorized}, we need not consider $\uv=\vv$ or $\xv=\yv$ when evaluating $\Suvxy$.
Suppose $\uv=\xv$.
Then neither $\vv$ nor $\yv$ can be $\xv$, so we need only consider the case of $n_\uv^2 n_\vv n_\yv$ with $\uv,\vv,\yv$ distinct,
and the case of $n_\uv^2 n_\vv^2$ with $\uv\neq \vv$.
In either case we need to consider $(n_\uv/n)^2$ as a factor in \eqref{eq:Suvxy-def}.
Let's again look at the example where $\uv,\vv,\xv,\yv$ are in $\Ac$, and $\Ac$ has 8 elements as above. Then for $t\le T$ with $T$ fixed and $n$ large, the relevant term in $\Suvxy$ is
\begin{align} \label{eq:Suvxy-example}
s_{\uv,\vv,\xv,\yv}(t,n) =  \sum_{t_1+t_2+t_3+t_4=t} ~\sum_{n_j + n_{\bar{j}}  = t_j}
 &\binom{n}{t} \binom{t}{t_1, t_2, t_3, t_4} \nonumber \\
\times&\binom{t_1}{n_1, n_{\bar1}} \gf{1}{2} \gf{1}{3} \gf{1}{4} \Xf{1} \nonumber\\
\times&\binom{t_2}{n_2, n_{\bar2}} \gf{2}{3} \gf{2}{4} \Xf{2} \nonumber\\
\times&\binom{t_3}{n_3, n_{\bar3}} \gf{3}{4} \Xf{3} \nonumber \\
\times&\binom{t_4}{n_4, n_{\bar4}} \Xf{4} \nonumber\\
\times& \Big(\frac{n_\uv}{n}\Big)\Big( \frac{n_\vv}{n}\Big)\Big( \frac{n_\xv}{n}\Big)\Big( \frac{n_\yv}{n}\Big) \Big(1+O\Big(\frac{1}{n}\Big)\Big).
\end{align}
We can use the identity
\begin{align} \label{eq:binom-ident2}
\sum_{n_j + n_{\bar j} = t_j} n_j^2 \binom{t_j}{n_j n_{\bar{j}}} D_j^{n_j} D_{\bar j}^{n_{\bar j}} 
= t_j D_j (D_j + D_{\bar j})^{t_j-1}
+ t_j(t_j-1)  D_j^2 (D_j + D_{\bar j})^{t_j-2} ~.
\end{align}
Let us apply this to say the sum over $n_3$ and $n_{\bar{3}}$ in \eqref{eq:Suvxy-example} with $\uv=\xv=3$, and we get
\begin{align}
\frac{1}{n} \frac{1}{t_3!}t_3 X_3\gt{1}{3} \gt{2}{3} (K_{3,1}t_1 + K_{3,2} t_2)^{t_3-1}  
+
\frac{1}{t_3!}t_3 (t_3-1) X_3^2 (\gt{1}{3} \gt{2}{3})^2 (K_{3,1}t_1 + K_{3,2} t_2)^{t_3-2}.
\end{align}
Note the first term gets factors of $n$ as $n^{t_3} (1/n)^{t_3-1}(1/n^2) = 1/n$, whereas the second term gets factors of $n$ as $n^{t_3}(1/n)^{t_3-2}(1/n^2)=1$.
Keeping in mind that $\gt{1}{3}\gt{2}{3} = 1+O(1/n)$, to leading order in $1/n$,  we have 
\begin{align}
\frac{1}{t_3!}t_3 (t_3-1) X_3^2 (K_{3,1}t_1 + K_{3,2} t_2)^{t_3-2}.
\end{align}
Let $t_3' = t_3 - 2$, and we get
\begin{align}
\frac{1}{t'_3!} X_3^2 (K_{3,1}t_1 + K_{3,2} t_2)^{t'_3}
\end{align}
and we set $t_3 = t'_3 + 2$ in all other places that $t_3$ appears.
The evaluation of the sum proceeds similarly as the case of $\Suv$, except now \eqref{eq:Kbar-def} is replaced by $\bar{K}_j= K_{j,\uv} + K_{j,\vv} + K_{j,\xv} + K_{j,\yv}$.
When $\uv=\xv$, we have $2K_{j,\uv}$ appearing in $\bar{K}_j$.
This effectively sends $Y_\uv$ to $Y_\uv^2$, and along with $X_\uv^2$ we get for the $n_\uv^2 n_\vv n_\yv$ case $W_\uv^2 W_\vv W_\yv$ as claimed.
The same argument can be worked out for the $n_\uv^2 n_\vv^2$ case, in which case we get $W_\uv^2 W_\vv^2$.

\subsection{Examples for $p=1$ and $p=2$}
Let us see how our formalism works to produce explicit formulas for $\lim_{n\to\infty}\EV_J[\bgbbraket{C/n}]$ by looking at how they apply at $p=1$ and $p=2$.

For $p=1$, the set of configurations $A$ contains just 4 elements as discussed in Section~\ref{sec:p=1}.
Then, $A_{1} = \{+-, -+\}$, and $A_2=A_{p+1} = \{++,--\}$.
It is easy to check all elements of $A$ play well with each other, so $\Delta_{\bv,\cv}=0$ for all $\bv,\cv\in A$.
Then from the formula for $W_\uv$ in \eqref{eq:W-formula} we see that $W_{\uv}=X_\uv$ for all $\uv\in A$.
Then
\begin{align}
W_{++} &= X_{++} =  Q_{++} F_{++} = Q_{++} = \cos^2\beta_1, \nonumber \\
W_{--} &= X_{--} = Q_{--} F_{--} = Q_{--} = \sin^2\beta_1, \nonumber \\
W_{+-} &= X_{+-} =  Q_{+-} F_{+-} = -i\sin\beta_1 \cos\beta_1  e^{-2\gamma_1^2}, \nonumber \\
W_{-+} &= X_{-+} =  Q_{-+} F_{-+}  = i\sin\beta_1 \cos\beta_1  e^{-2\gamma_1^2}.
\end{align}
Furthermore, note that $\av^* =\av$ when $p=1$. So combining \eqref{eq:M1}, \eqref{eq:factorized}, and \eqref{eq:Suv}, we have
\begin{align}
\lim_{n\to\infty} \EV[\bgbbraket{C/n}] &= \frac{i}{2}\gamma_1 \left[\sum_{\uv\in A}(u_1 + u_{-1}) W_\uv\right] \left[\sum_{\vv\in A} (v_{1} - v_{-1}) W_\vv \right] \nonumber \\
&= \frac{i}{2} \gamma_1 (2W_{++} - 2W_{--}) (2W_{+-} - 2W_{-+}) \nonumber \\
&= \gamma_1 e^{-2\gamma_1^2} \sin4\beta_1 
\end{align}
as we have already shown in Section~\ref{sec:p=1}.

Now, let us consider the example of $p=2$, where the set of configurations $A$ has $16$ elements.
We first list $W_\uv$ for $\uv\in A_{p+1} = \{++++,+--+,-++-,----\}$, which are
\begin{align}
W_{++++} &= Q_{++++} = \cos^2\beta_1 \cos^2\beta_2, \nonumber \\
W_{+--+} &= Q_{+--+} = \cos^2\beta_1 \sin^2\beta_2, \nonumber \\
W_{-++-} &= Q_{-++-} = \sin^2\beta_1 \cos^2\beta_2, \nonumber \\
W_{----} &= Q_{----} = \sin^2\beta_1 \sin^2\beta_2.
\end{align}
Now, we need to compute $W_\uv$ for the remaining 12 elements, which constitute the set $\Ac = A\setminus A_{p+1}$. We first bipartition $B=D\cup D^c$ such that if $\bv\in D$ then $\bar\bv \in D^c$.
We also assign indices to configurations of $D$ so that if two configurations have indices $j\le k$ then $j$ plays well with $k$.
For concreteness, we can choose the convention where all configurations of $D$ are chosen so that their first $p$ elements have even parity (i.e., $b_1b_2\cdots b_p=1$). Then for $p=2$, we have
\begin{align} \label{eq:Ddef}
D = \{&1=++--, \quad 2=--++, \quad 3=++-+, \nonumber \\
 & 4=--+-, \quad 5=+++-, \quad 6=---+\}.
\end{align}
Note that $1,2,3,4$ are in $A_1$, whereas 5 and 6 are in $A_2$.
Now we need to consider elements of $D$ that don't play well with others.
After some inspection, the only cases we need to worry about are that 5 and 6 do not play well with 1 and 2.
Then the only non-zero $\Delta_{j,i}$ are
\begin{equation}
\Delta_{6,1} = -\Delta_{6,2} = -\Delta_{5,1} = \Delta_{5,2}= 8 \gamma_1 \gamma_2 .
\end{equation}
Let us define
\begin{align}
\Omega &= 4 \gamma_1 \gamma_2 \cos2\beta_1, \\
\Lambda &=  4  \gamma _1 \gamma _2 \sin 2 \beta _1\exp({-2 \gamma _1^2}),
\end{align}
as shorthands.
We then find all the $W_\uv$'s by applying \eqref{eq:W-formula}, which yields
\begin{align}
W_6 &= X_6 = \frac{i}{2} e^{-2\gamma_1^2}\sin2\beta_1 \sin^2\beta_2, \nonumber \\
W_5 &= X_5 = -\frac{i}{2} e^{-2\gamma_1^2}\sin2\beta_1 \cos^2\beta_2, \nonumber \\
W_4 &= X_4 = \frac{i}{2} e^{-2(\gamma_1^2+\gamma_2^2) + \Omega} \sin^2\beta_1 \sin 2\beta_2, \nonumber \\
W_3 &= X_3 = -\frac{i}{2} e^{-2(\gamma_1^2+\gamma_2^2) - \Omega} \cos^2\beta_1 \sin 2\beta_2, \nonumber \\
W_2 &= X_2 e^{W_6 \Delta_{6,2} + W_5 \Delta_{5,2}} = -\frac14 e^{-2\gamma_2^2 - i\Lambda} \sin 2\beta_1 \sin2\beta_2,  \nonumber \\
W_1 &= X_1 e^{W_6 \Delta_{6,1} + W_5 \Delta_{5,1}} = -\frac14 e^{-2\gamma_2^2+ i\Lambda} \sin 2\beta_1 \sin2\beta_2.
\end{align}
Note that $W_\uv$ for the other elements $\uv\in D^c$ are given via  $W_{\bar\uv} = -W_{\uv}$.
Thus, for $p=2$, we have
\begin{align}
\lim_{n\to\infty} \EV[\bgbbraket{C/n}] &= \frac{i}{2}\sum_{r=1}^2 \gamma_r \left[\sum_{\uv\in A}(u_r^* + u_{-r}^*) W_\uv\right] \left[\sum_{\vv\in A} (v_{r}^* - v_{-r}^*) W_\vv \right] \nonumber \\
&= \left[\gamma_1 \Gamma_1(\vparam)+ \gamma_2 \Gamma_2(\vparam)\right] \exp\big[{-2 (\gamma_1^2 + \gamma _2^2)}\big],
 \label{eq:p=2}
\end{align}
where
\begin{align}
\Gamma_1(\vparam) &=
2\Big(e^{2\gamma_2^2} \sin2\beta_1 \cos 2\beta_2 + \sin 2\beta_2 (\cos 2\beta_1 \cosh \Omega -  \sinh \Omega)\Big)   \nonumber \\
   &\quad \times 
   \left( \cos 2 \beta _1 \cos 2 \beta _2 
		- e^{-2 \gamma _2^2} \sin 2 \beta _1 \sin 2 \beta _2
   \cos \Lambda \right), \\
\Gamma_2(\vparam) &=
\left(\cosh \Omega + e^{2\gamma_1^2} \sin 2\beta_1\sin \Lambda - \cos2\beta_1 \sinh \Omega\right)\sin 4\beta_2.
\end{align}

%
%To see this, let us denote $W_a = W_{++++}$, $W_{b} = W_{+--+}$, $W_c = W_{-++-}$, and $W_d= W_{----}$.
%And using our index notation for the elements of $D$ given in \eqref{eq:Ddef}, then \eqref{eq:p=2} is explicitly
%\begin{align*}
% && \quad 2i\gamma_1
%	& (W_{++++} + W_{++--} + W_{----} + W_{--++} - W_{-+-+} - W_{-++-} - W_{+--+} - W_{+-+-}) \\
%&&	\times~& (W_{+++-} + W_{++-+} + W_{---+} + W_{--+-} - W_{-+--} - W_{-+++} - W_{+---} - W_{+-++}) \\
%&&+ ~ 2i\gamma_2
%	& (W_{++++} + W_{+++-} + W_{-++-} + W_{-+++} - W_{----} - W_{---+} - W_{+---} - W_{+--+}) \\
%&&	\times~& (W_{++-+} + W_{++--} + W_{-+--} + W_{-+-+} - W_{--+-} - W_{--++} - W_{+-+-} - W_{+-++}) \\
%&&=\quad 2i\gamma_1
%	& (W_{a} + W_{1} + W_{d} + W_{2} - W_{\bar2} - W_{c} - W_{b} - W_{\bar1})  (W_{5} + W_{3} + W_{6} + W_{4} - W_{\bar4} - W_{\bar5} - W_{\bar6} - W_{\bar3}) \\
%&&+ ~ 2i\gamma_2
%	& (W_{a} + W_{5} + W_{c} + W_{\bar5} - W_{d} - W_{6} - W_{\bar6} - W_{b}) (W_{3} + W_{1} + W_{\bar4} + W_{\bar2} - W_{4} - W_{2} - W_{\bar1} - W_{\bar3}) \\
%&&= \quad 2i\gamma_1
%& (W_a + W_d - W_b - W_c + 2W_1 + 2W_2)(2W_3 + 2W_4 + 2W_5 + 2W_6) \\
%&& + ~ 2i\gamma_2 &( W_a + W_c - W_b - W_d)(2W_1 + 2W_3 - 2W_2 - 2W_4).
%\end{align*}
%Plugging in values of $W_i$ gives the expression above.

\vspace{-4pt}
\section{Discussion\label{sec:discussion}}
%\vspace{-1pt}

The QAOA is a general purpose quantum optimizer whose computational power has not been fully explored.  Early applications were to problems such as MaxCut on bounded degree graphs.  Here it is known that for fixed depth $p$, as the graph size increases there can be upper bounds on how well the algorithm can perform.
For example, for MaxCut on 3-regular bipartite graphs (i.e., completely satisfiable) with $o(n)$ subgraphs that are triangles, squares or pentagons,
the achieved approximation ratio at $p=2$ is 0.7559~\cite{QAOA}. (Poor reader, don't think this has anything to do with the Parisi constant of 0.763\ldots in  \eqref{eq:Parisi}.)
For fixed degree graphs we need $p$ to increase at least logarithmically so that each edge sees the whole graph. What we mean by ``sees'' is this:  Consider a cost function that is a sum of local terms described in terms of a graph problem. Pick one term say $C_\alpha$ corresponding to clause $\alpha$.  Then in evaluating \eqref{eq:QAOA_obj} we look at 
\begin{equation}
	U^\dag(C, \gamma_1)\cdots U^\dagger(B, \beta_p)  C_\alpha  U(B, \beta_p) \cdots U(C, \gamma_1)
\end{equation}
as in \eqref{eq:wavefunction}.
The locality of this operator is found by taking each qubit in $C_\alpha$ and moving out a distance $p$ on the instance graph to see which other qubits are involved.  If $p$ does not grow with $n$, then for large $n$, not all qubits are ``seen'' by the clause. Without seeing the whole graph the algorithm can not detect contradictions from say large loops.
Therefore, to have hope of success on bounded degree graphs, we need $p$ to grow at least as a big enough constant times $\log(n)$.  For near term quantum computers with even 1000 qubits, this is easily achievable since for say 3-regular graphs at $p=8$, each clause is guaranteed to see the whole graph. 

For problems on graphs whose degree grows with the number of bits, the previous arguments do not apply. For the Sherrington-Kirkpatrick model, the associated graph is the complete graph.
Therefore, each clause sees all the qubits at $p=1$, and sees all the qubits as well as all the edges at $p=2$.
Hence, one may imagine that for large fixed $p$, that in the large $n$ limit, the QAOA will perform well.
We have explored this prospect in this work by obtaining a formula for the expected value of the cost function for an arbitrary QAOA state in the limit as $n\to\infty$. 
The complexity of our formula grows as $O(16^p)$, so without too much trouble we could  optimize up to $p=8$ and evaluate up to $p=12$, and our results are shown in Figure~\ref{fig:Cexp}. 
(A recent result in Ref.~\cite{basso2022quantum} improves the complexity to $O(p^24^p)$ which enables evaluation up to $p=20$.)
It is currently not possible to know for certain if for large $p$ the performance asymptotes to the Parisi value \eqref{eq:Parisi} or something less optimal.
At $p=4$, we have crossed the energy at the phase transition which is  $-0.5 n$.
At $p=11$, we have surpassed the performance of the spectral relaxation and the standard semi-definite programming algorithms which yield $-2n/\pi$.
%It appears possible from fitting that the performance is described by a power-law convergence to the lowest energy.
Nevertheless, as we discussed in Section~\ref{sec:SK}, the recent classical algorithm by Montanari~\cite{Montanari2018} can efficiently find a string whose energy is between $(1-\epsilon)$ and 1 times the lowest energy, assuming a widely believed conjecture that the SK model has ``full replica symmetry breaking.''
Hence, we can only hope to match this with the QAOA.
We want to improve our techniques to determine the asymptotic behavior of the QAOA's performance.
Furthermore, we can apply our techniques to other problems where we average over instances. 
For example, we can consider generalizations of the SK model where multi-spin couplings are allowed. Since some of these models are shown to have no ``full replica symmetry breaking'' \cite{CGPR19}, it is believed that classical algorithms like Montanari's will fail to find near-optimal solutions~\cite{GJ19,AM20}.
It would be interesting to see how the QAOA performs for those problems, and whether it will outperform the classical algorithms. %encounter similar bottlenecks.
We also imagine the possibility of extending our techniques to $p$ growing with $n$.

\section*{Acknowledgement}
We are grateful to Chris Baldwin, Mehran Kardar, Chris Laumann, Giorgio Parisi, and Shivaji Sondhi for helpful discussions. We thank the hospitality of the Galileo Galilei Institute where many discussions took place.
We also thank David Gamarnik for valuable discussions.
We thank four anonymous referees for their careful reading of the paper and their suggestions, which have improved it.
E.F.~thanks Hartmut Neven for discussion and encouragement.
L.Z.~thanks the Google AI Quantum team for a summer internship where much of this work was done.
L.Z.~acknowledges partial support from the National Science Foundation.
This work was partially funded by ARO contract W911NF-17-1-0433.

\bibliographystyle{quantum}

\bibliography{QAOA_SK_refs}

\end{document}